\documentclass[12pt, draftclsnofoot, onecolumn]{IEEEtran}
\usepackage{floatflt}

\usepackage{color}
\usepackage{cite}

\usepackage[normalem]{ulem}
\usepackage{multirow}
\usepackage{graphicx}
\usepackage{amsmath}
\usepackage{amsthm}
\usepackage{amssymb}

\usepackage{verbatim}

\usepackage{psfrag,array}

\usepackage{subfigure}

\usepackage{stfloats}

\usepackage{amsmath}
\usepackage{epstopdf}
\usepackage{algorithmic}

\usepackage{booktabs}

\usepackage{url}

\usepackage{algorithm}
\usepackage{algorithmic}

\usepackage{algorithm}
\usepackage{algorithmic}

\newcommand{\p}[1]{\mathop{\mbox{\it p} } }

\renewcommand{\vec}[1]{\ensuremath{\boldsymbol{#1}}}
\newcommand{\be}{\begin{equation}}
\newcommand{\ee}{\end{equation}}
\newcommand{\ba}{\begin{array}}
\newcommand{\ea}{\end{array}}
\newcommand{\bea}{\begin{eqnarray}}
\newcommand{\eea}{\end{eqnarray}}
\newcommand{\bean}{\begin{eqnarray*}}
\newcommand{\eean}{\end{eqnarray*}}

\newcommand{\argmin}{\mathop{\arg\min}}

\newcommand{\rmh}{^{\dag}}
\newcommand{\rmt}{^{\rm T}}

\newcolumntype{L}[1]{>{\raggedright\let\newline\\\arraybackslash\hspace{0pt}}m{#1}}
\newcolumntype{C}[1]{>{\centering\let\newline\\\arraybackslash\hspace{0pt}}m{#1}}
\newcolumntype{R}[1]{>{\raggedleft\let\newline\\\arraybackslash\hspace{0pt}}m{#1}}

\definecolor{white}{rgb}{1,1,1}

\newtheorem{theorem}{Theorem}
\newtheorem{lemma}{Lemma}
\newtheorem{example}{Example}

\newtheorem{property}{Property}

\newtheorem{corollary}{Corollary}

\begin{document}

\title{A Generalized Zero-Forcing Precoder with Successive Dirty-Paper Coding in MISO Broadcast Channels}
\author
{
Sha Hu and Fredrik Rusek
\thanks{The authors are with the Department of Electrical and Information Technology, Lund University, Lund, Sweden (\{firstname.lastname\}@eit.lth.se).}
}

\maketitle

\vspace{-4mm}
\begin{abstract}
In this paper, we consider precoder designs for multiuser multiple-input-single-output (MISO) broadcasting channels. Instead of using a traditional linear zero-forcing (ZF) precoder, we propose a generalized ZF (GZF) precoder in conjunction with successive dirty-paper coding (DPC) for data-transmissions, namely, the GZF-DP precoder, where the suffix \lq{}DP\rq{} stands for \lq{}dirty-paper\rq{}. The GZF-DP precoder is designed to generate a band-shaped and lower-triangular effective channel $\vec{F}$ such that only the entries along the main diagonal and the $\nu$ first lower-diagonals can take non-zero values. Utilizing the successive DPC, the known non-causal inter-user interferences from the other (up to) $\nu$ users are canceled through successive encoding. We analyze optimal GZF-DP precoder designs both for sum-rate and minimum user-rate maximizations. Utilizing Lagrange multipliers, the optimal precoders for both cases are solved in closed-forms in relation to optimal power allocations. For the sum-rate maximization, the optimal power allocation can be found through water-filling, but with modified water-levels depending on the parameter $\nu$. While for the minimum user-rate maximization that measures the quality of the service (QoS), the optimal power allocation is directly solved in closed-form which also depends on $\nu$. Moreover, we propose two low-complexity user-ordering algorithms for the GZF-DP precoder designs for both maximizations, respectively. We show through numerical results that, the proposed GZF-DP precoder with a small $\nu$ ($\leq\!3$) renders significant rate increments compared to the previous precoder designs such as the linear ZF and user-grouping based DPC (UG-DP) precoders.
\end{abstract}


\section{Introduction}
In the emerging Internet of things (IoT) \cite{IoT} and device-to-device (D2D) \cite{HR161} communication systems, a transmit node equipped with $M$ transmit antennas may broadcast messages simultaneously to $N$ low-cost receive nodes that are equipped with a single antenna. Under the assumption that the number of transmit antennas are much larger than the number of served users, i.e., $M\!\gg\!N$, which is known as massive multiple-input-multiple-output (MIMO) systems\cite{5G1}, the multiple-input-single-output (MISO) broadcasting channels corresponding to different users that link the transmit and receive nodes are approximately orthogonal to each other. Consequently, the zero-forcing (ZF) precoders applied at the transmit nodes can efficiently eliminate the inter-user interference, and the MISO channels can be decomposed into a number of parallel and independent single-input-single-output (SISO) channels in such cases. 

In small-antenna systems such as small cells \cite{5G2} with compact base-stations and WiFi systems, however, compared to the number of served users the number of transmit antennas are usually limited. Further, in current 3GPP standard \cite{3gpp36201}, LTE-A systems support only up to 8 transmit antennas. Although future releases may support massive-MIMO or full-dimension MIMO (FD-MIMO)\cite{FDMIMO} and the number of transmit antennas at the eNode-B may increase to 64 for 2-D antenna array designs, the intended number of served users will also increase due to the vast connections featured in 5G systems. Consider the case where $N$ is comparable to $M$, in order to fully eliminate the inter-user interference, the linear ZF precoder performs poorly due to the non-orthogonality of the MISO broadcast channel vectors. Therefore, advanced precoder designs are required to improve the transmit power-efficiency and increase the rates of data-transmissions. 

Some of the typical precoder designs are to preserve parts of the inter-user interference and mitigate them with the techniques of channel coding with side information (CCSI). CCSI has generated much research interests due to its applications in data hiding \cite{PA99}, precoding for interference channels \cite{CS03}, and transmitter cooperation in Ad-hoc networks\cite{M06}. Gelfand and Pinsker in \cite{GP80} derive the capacity of a single-user memoryless channel with an additive interference signal $\vec{s}$ known to the transmitter, but not the receiver. Consider a received signal
\bea \label{md0} \vec{y}=\vec{x}+\vec{s}+\vec{z}, \eea
where $\vec{x}$, $\vec{y}$ are transmit and receive signals, and $\vec{z}$ is the unknown Gaussian noise, respectively. The capacity of model (\ref{md0}) is shown to equal
\bea \label{GP} \mathcal{C}=\max_{p\left(\vec{u}, \vec{x}|\vec{s}\right)}\left\{I(\vec{u};\vec{y})-I(\vec{u};\vec{s})\right\} .  \eea
where $\vec{u}$ is an auxiliary random variable and the maximum is taken over all joint probability distributions. Based on the result (\ref{GP}), Costa shows in \cite{C83} that with dirty-paper coding (DPC), the channel capacity $ \mathcal{C}$ is the same even if the interference $\vec{s}$ is not present. Utilizing the same principle, the DPC scheme can be extended to multi-user Gaussian vector broadcast channels\cite{GK11}, and DPC capacity regions have been derived via the uplink-downlink duality between broadcast channels and multiple-access channels \cite{VT03, JV04}. Practical DPC designs based on finite-alphabets have been extensively developed such as Tomlinson-Harashima precoding \cite{SL09}, Lattice Precoding \cite{ER05}, and trellis coded quantization and modulation \cite{ES05, YV05}.

Caire and Shamai in \cite{CS03} propose a ZF based DPC (ZF-DP) design for MISO broadcast channels. They show that with successive DPC utilized at transmitter, the sum-rate of the ZF-DP precoder is close to the optimal DPC. In \cite{DL07}, the authors propose a successive ZF-DP (SZF-DP) precoding scheme and show that in the low SNR regime, the SZF-DP has similar performance as a successive ZF (SZF) precoder, where the SZF-DP and SZF precoders are direct extensions of the ZF-DP and linear ZF precoders in \cite{CS03} for MIMO broadcast channels. In \cite{TJBB131, TJBB132} the authors further extend the ZF-DP and SZF-DP precoders subject to per-antenna power constraint (PAPC) instead of a sum-power constraint (SPC). Nevertheless, all the successive DPC based precoder designs in \cite{CS03, DL07, TJBB131, TJBB132} assume a full successive DPC scheme. As the number of users $N$ increases, the successive DPC becomes prohibitive as it needs to consider the inter-user interference up to $N\!-\!1$ users. Recently, the authors in \cite{ML16} propose a user-group based DPC precoder (UG-DP), which splits the $N$ users into $g$ disjoint groups with each group containing $N_g$ users\footnote{For notational convenience, we assume that $N$ is divisible by $g$ and let $N_g\!=\!N/g$. But it can be straightforwardly modified to other cases with minor changes.}. The inter-group interferences are eliminated by the precoder, while the intra-group interferences are canceled with successive DPC that is implemented on each user-group independently. With a small $N_g$, the DPC has less-complexity and is feasible \cite{SL09, ER05, ES05, YV05, ML16, LJ07}. However, as different user-groups are orthogonalized to each other, the UG-DP also suffers from rate-losses, especially when the channel vectors of different user-groups are spatially correlated.

In this work, we propose a generalized ZF precoder (GZF) design in conjunction with successive DPC, namely, the GZF-DP precoder, which unifies the designs of the UG-DP and the ZF-DP precoders. Instead of considering $N\!-\!1$ users in previous designs, we consider inter-user interference up to $\nu$ users, where the parameter $\nu$ is up to design and provides a trade-off between the rates and implementation complexity of the successive DPC\footnote{Instead of using DPC at transmitter, in cooperative networks \cite{NH04} the receiver nodes can implement successive interference cancellations (SIC) to achieve the same rates as the DPC. However, that requires a cost of communicating between the receive nodes. In which case, the parameter $\nu$ represents a maximal number of communication channels needed for the receive nodes.}. By setting $\nu\!=\!0$, the GZF-DP precoder degrades to the linear ZF precoder, which has low complexity (no DPC is needed) but also low rates. On the other hand, with setting $\nu\!=\!N\!-\!1$, the GZF-DP precoder is identical to the ZF-DP precoder \cite{CS03}, which performs better than the other settings of $\nu$ but also has the highest DPC implementation complexity. Moreover, as the UG-DP precoder can be viewed as a special case of the GZF-DP precoder, it renders lower rates than the GZF-DP precoder with $\nu\!=\!N_g\!-\!1$. 

With the GZF-DP precoder, we consider two optimal designs: sum-rate maximization and minimum user-rate maximization, that are aiming to maximize the overall throughput and the quality of service (QoS), respectively. Using Lagrange multipliers, the optimal GZF-DP precoder designs for both cases are found in closed-form which depend on optimal power-allocations. For the sum-rate maximization, the optimal power allocation is found through a water-filling scheme in relation to modified water-levels introduced by preserving the inter-user interference up to $\nu$ users. While for the minimum user-rate maximization, the optimal power allocation can be solved directly in closed-form which also depends on $\nu$. Moreover, we provide two low-complexity algorithms for optimal user-orderings for both maximizations, respectively. We show through numerical results that, the proposed GZF-DP precoder is superior to the previous ZF and UG-DP precoders, and most interestingly, with a small value of $\nu$ ($\leq \!3$) the proposed GZF-DP precoder performs close to the ZF-DP precoder\cite{JV04}, i.e., the GZF-DP precoder with $\nu\!=\!N\!- \!1$.

Notice that, as the precoder designs in \cite{DL07, TJBB131, TJBB132} follow similar approaches as those in \cite{CS03}, the proposed GZF-DP precoder can also be extended to MIMO broadcast channels and PAPC constraint, which is a generalization of the SZF-DP precoder by only performing DPC up to $\nu$ multiple-receive-antenna users. However, as in \cite{DL07, TJBB131, TJBB132} only the sum-rate maximization with a full DPC is considered, an interesting fact that the sum-rate maximization actually sacrifices the user-rates of some of the last users (corresponding to the last columns of channel matrix $\vec{H}$) compared to the linear ZF precoder is not shown. With the variable $\nu$ increasing from 0 to $N\!-\!1$, this property is clear shown in this work, which also motivates us to consider the minimum user-rate maximization for the proposed GZF-DP precoder.

The rest of the paper are organized as follows. In Sec. II, we briefly introduce the MISO system model and the previous precoder designs. In Sec. III, we elaborate the proposed GZF-DP precode designs in detail for sum-rate and minimum user-rate maximizations, respectively. We also analyze the low-complexity ordering algorithms for both maximization problems. Empirical results are provided in Sec. IV, and Sec. V summarizes the paper.

\subsection*{Notations:}
Throughout this paper, superscripts $(\cdot)^{-1}$, $(\cdot)^{1/2}$, $(\cdot)^{\ast}$,
$(\cdot)\rmt$ and $(\cdot)\rmh$ stand for the inverse, matrix square root, complex conjugate, transpose, and
Hermitian transpose, respectively. Boldface letters indicate vectors and boldface uppercase letters designate matrices. We also reserve $a_{m,n}$ to denote the element at the $m$th row and $n$th column of matrix $\vec{A}$, $a_{m}$ to denote the $m$th element of vector $\vec{a}$, and $\vec{I}$ to represent the identity matrix. The operators $\mathcal{R}\{\cdot\}$ and $\mathrm{Tr}(\cdot)$ take the real part and the trace of the arguments, and $[\cdot]^{+}$ is the non-negative protection. In addition, $\mathcal{J}_1\!\setminus\mathcal{J}_2$ returns a set that contains all elements in set $\mathcal{J}_1$ that are not in $\mathcal{J}_2$, and the expressions $\vec{A}\!\succ\!\vec{B}$ and $\vec{A}\!\succeq\!\vec{B}$ represent that $\left(\vec{A}\!-\!\vec{B}\right)$ is positive definite and semi-positive definite, respectively.

\section{System Model and Previous Sum-rate Maximization Precoder Designs}

Consider an MISO system with an $M$-antenna transmitter and $N$ single-antenna users with assumption $M\!\geq\!N$. The channel vector from the transmitter to the $n$th user is denoted as $\vec{h}_n\!\in\!\mathbb{C}^{M \times 1}$, and the $m$th entry $h_{mn}$ of $\vec{h}_n$ is the channel gain from the $m$th transmit antenna to the $n$th user. Denote the $N\!\times\!M$ channel
\bea \vec{H}=[\vec{h}_1 \; \vec{h}_2 \;\ldots \;\vec{h}_N]\rmt, \eea
and let the $N\!\times\!1$ vectors
{\setlength\arraycolsep{2pt}  \bea \vec{y}&=&[y_1 \; y_2 \;\ldots \; y_N]\rmt, \notag \\
  \vec{x}&=&[x_1 \; x_2 \;\ldots \; x_N]\rmt, \notag \\
 \vec{z}&=&[z_1 \; z_2 \;\ldots \; z_N]\rmt, \eea}
\hspace{-1.4mm}where $x_n$ is the DPC-encoded symbol of the $n$th user that cancels the non-causal interference from the other users, and $y_n$, $z_n$ is the received sample and the noise term corresponding to the $n$th user, respectively. With an $M\!\times\!N$ precoding matrix $\vec{P}$ applied at the transmitter, the received signals at the $N$ autonomous users can be compactly written as 
\bea \label{md1} \vec{y} =\vec{H}\vec{P}\vec{x}+\vec{z},\eea
where the noise term $\vec{z}$ comprises identical and independently distributed (IID) complex Gaussian variables with zero mean and a covariance matrix $N_0\vec{I}$. The transmit symbols $x_n$ are uncorrelated due to DPC encoding and have unit-transmit power, that is, $\mathbb{E}[\vec{x}\vec{x}\rmh]\!=\!\vec{I}$. In addition, the transmit node is subject to a total transmit power constraint $P_{\mathrm{T}}$ such that
\bea \label{con1} \mathrm{Tr}\left(\vec{P}\vec{P}\rmh\right)\leq P_{\mathrm{T}}.\eea

\subsection{Optimal DPC Precoder}
Denote the effective channel $\vec{F}\!=\!\vec{H}\vec{P}$, the interference channel corresponding to each of the $N$ users from (\ref{md1}) can be written as
\bea   y_n= f_{n,n}x_n+\sum_{k=1}^{n-1}f_{n,k}x_k+\sum_{k=n+1}^{N}f_{n,k}x_k+z_n.   \eea
With a successive DPC \cite{C83} encoding scheme, the interference term $\sum\limits_{k=1}^{n-1}f_{n,k}x_k$ is non-causally known and canceled, while the causal interference term $\sum\limits_{k=n+1}^{N}f_{n,k}x_k$ is regarded as additive noise. Therefore, the optimal DPC precoder that maximizes the sum-rate is designed by solving the following problem
{\setlength\arraycolsep{2pt} \bea \label{DP} &&\underset{\vec{F}}{\mathrm{maximize\;\;}} \sum_{n=1}^{N}\log\left(1+\frac{|f_{n,n}|^2}{N_0+\sum\limits_{k=n+1}^{N}|f_{n,k}|^2}\right) \notag   \\
&& \mathrm{subject\; to\;\;}(\ref{con1}).  \eea}
\hspace{-1.4mm}Directly optimizing (\ref{DP}) is computationally complex as it is a non-convex problem. In \cite{JV05} the authors propose an iterative water-filling scheme to solve (\ref{DP}) based on the uplink-downlink duality. Although the optimal DPC precoder achieves the capacity region \cite{WY06} of the multi-user MISO broadcast channels, the linear ZF precoder is widely used due to its simple implementation.

\subsection{Linear ZF Precoder}
The linear ZF precoder is set to
\bea \label{zf} \vec{P}=\vec{H}\rmh\left(\vec{H}\vec{H}\rmh\right)^{-1}\vec{F},\eea
where $\vec{F}$ is an $N\!\times\!N$ diagonal matrix. With (\ref{zf}), the constraint (\ref{con1}) changes to
\bea  \label{con2} \mathrm{Tr}\left(\vec{F}\rmh\left(\vec{H}\vec{H}\rmh\right)^{-1}\vec{F}\right)\leq P_{\mathrm{T}}.\eea
Denote $\vec{G}\!=\!\left(\vec{H}\vec{H}\rmh\right)^{\!-1}$, the sum-rate maximization for linear ZF precoder is then formulated as
{\setlength\arraycolsep{2pt} \bea \label{prbm1} &&\underset{f_{n,n}}{\mathrm{maximize\;\;}}  R=\sum_{n=1}^{N}\log\left(1+\frac{|f_{n,n}|^2}{N_0}\right) \notag   \\
&&\mathrm{subject\; to\;\;}   \sum_{n=1}^{N}g_{n,n}|f_{n,n}|^2\leq P_{\mathrm{T}}.\eea}
\hspace{-1.4mm}The optimal power allocation is found through the water-filling scheme,
\bea  \label{wf1} |f_{n,n}|^2=N_0\left[\frac{1}{\lambda g_{n,n}}-1\right]^{+}\!,\eea
where $\lambda\!\geq\!0$ is a constant such that power constraint (\ref{con2}) is satisfied. The optimal sum-rate reads
\bea  \label{sumrate1} R^{\mathrm{sum}}=\sum_{n=1}^{N}\big[-\log\left( \lambda g_{n,n}\right)\big]^{+},\eea
As the linear ZF precoder completely eliminates the inter-user interference, it results in low transmit power-efficiencies (even with regularizations\cite{ZF2}), especially when $\vec{H}$ is ill-conditioned. In \cite{CS03}, the authors propose a ZF-DP precoder that only nulls out the causal inter-user interference through ZF, and utilize successive DPC to cancel the non-causal interference.

\subsection{ZF-DP Precoder}

Assuming the channel decomposition $\vec{H}\!=\!\vec{R}\vec{U}$, where $\vec{R}$ is an $N\!\times\!N$ lower-triangular matrix and $\vec{U}$ is an $N\!\times\!M$ unitary matrix, the ZF-DP precoder is set to $\vec{P}\!=\!\vec{U}\rmh\vec{B}$, and the $N\!\times\!N$ diagonal matrix $\vec{B}$ represents the power allocation whose $n$th diagonal element is $b_n$. The effective channel with the ZF-DP precoder equals $\vec{F}\!=\!\vec{R}\vec{B}$, and the received sample $y_n$ reads
\bea   y_n= f_{n,n}x_n+\sum_{k=1}^{n-1}f_{n,k}x_k+z_n.   \eea
Through successive DPC encoding, the non-casual interference $\sum\limits_{k=1}^{n-1}f_{n,k}x_k$ is nulled out for each of the users, and the sum-rate maximization problem can be formulated as
{\setlength\arraycolsep{2pt} \bea \label{prbm2} &&\underset{b_n}{\mathrm{maximize\;\;}}  \sum_{n=1}^{N}\log\left(1+\frac{|b_n r_{n,n}|^2}{N_0}\right) \notag   \\
&&\mathrm{subject\; to\;\;}  \sum_{n=1}^{N}b_n^2\leq P_{\mathrm{T}}.\eea}
\hspace{-1.4mm}The optimal power allocation $b_n$ can also be found through standard water-filling. Although the ZF-DP precoder renders promising performance, the implementation of successive DPC becomes over complex when $N$ is large. To reduce the DPC complexity, the authors in \cite{ML16} propose a low-complexity UG-DP precoder. 

\subsection{UG-DP Precoder}
We next briefly introduce the UG-DP precoder design. Assuming the same channel decomposition as with ZF-DP precoder, but now we constrain $\vec{R}$ to be block-diagonal, with each block $\vec{R}_k$ ($1\!\leq\!k\!\!\leq\!g$) being an $N_g\!\times\!N_g$ lower-triangular matrix. Let the $N_g\!\times\!M$ sub-matrix $\vec{H}_k$ comprise the row vectors in $\vec{H}$ corresponding to the users in the $k$th group, and the $(N\!-\!N_g)\!\times\!M$ sub-matrix $\bar{\vec{H}}_k$ comprise the remaining row vectors. With decomposition $\vec{U}\!=\!\left[\vec{U}_1,\vec{U}_2,\cdots,\vec{U}_g\right]\rmh$, each $M\!\times\!N_g$ component $\vec{U}_k$ can be obtained through
\bea \label{QR2}\vec{H}_k\left(\vec{I}-\bar{\vec{H}}_k\rmh\left(\bar{\vec{H}}_k\bar{\vec{H}}_k\rmh\right)^{-1}\bar{\vec{H}}_k\right)\!=\!\vec{R}_k\vec{U}_k\rmh. \eea
Then, with the matrix $\vec{U}_k$ calculated via (\ref{QR2}), the optimal $\vec{P}$ equals $\vec{P}\!=\!\vec{U}\rmh\vec{B}$ and the effective channel becomes $\vec{F}\!=\!\vec{R}\vec{B}$, where the diagonal matrix $\vec{B}$ represents the power allocation to different users. Then, the remaining processes follow the ZF-DP precoder design. Although the UG-DP precoder reduces the complexity of DPC by user-grouping, it also suffers from rate-losses from the orthogonalization of different user-groups. In order to increase the rates of the UG-DP precoder while keeping a similar complexity, we can extend the block-diagonal lower-triangular $\vec{R}$ and $\vec{F}$ to be band-shaped matrices. That is, the connections among different user-groups are preserved such that, only the elements along the main diagonal and the first $\!N_g\!-\!1$ lower-diagonals of $\vec{R}$ and $\vec{F}$ can take non-zero values. The proposed GZF-DP precoder design is based on such a principle and is explained in detail next.

\section{Optimal Designs of the Proposed GZF-DP Precoder}
Instead of assuming $\vec{F}$ to be diagonal or block-diagonal such as in previous designs, we let $\vec{F}$ to be a band-shaped and lower-triangular for the GZF-DP precoder design,
{\setlength\arraycolsep{1pt} \bea \label{Fmat} \vec{F}\!=\!\left[\!\begin{array}{ccccccc}
f_{1,1}&~&~&~&~&~&~\\ f_{2,1}& f_{2,2}&~&~&~&~&~\\  
\vdots&f_{3,2}&\ddots&~&~&~&~ \\ 
f_{\nu+1,1}&\vdots&\ddots&\ddots&~&~&~\\~&f_{\nu+2,2}&\ddots&\ddots&\ddots&~&~\\
~&~&\ddots&\ddots&\ddots&\ddots&~\\ ~&~&~&f_{N, N-\nu}&\cdots&f_{N, N-1}&f_{N, N}
\end{array} \!\right]\!. \;\;\eea}
\hspace{-1.4mm}The parameter $\nu$ denotes the interfering depth of the effective MISO broadcasting channels. For simpler descriptions, we define two operations as
\vspace{-4mm}
{\setlength\arraycolsep{2pt}  \bea n\ominus \nu&=&\max(n-\nu, 0), \notag \\
n\boxplus\nu&=&\min(n+\nu, N).   \eea}
\hspace{-1.4mm}The GZF-DP precoder generalizes the linear ZF precoder in the sense that $\nu$ can be set larger than 0. Under the case $\nu\!=\!0$, the GZF-DP precoder degrades to the linear ZF precoder and no DPC is needed. With $\vec{F}$ defined in (\ref{Fmat}), the received sample $y_n$ of the $n$th user reads
\bea  \label{md3} y_n= f_{n,n}x_n+\sum_{k=n\ominus\nu}^{n-1}f_{n,k}x_k+z_n.  \eea
As the interference $\sum\limits_{k=n\ominus\nu}^{n-1}f_{n,k}x_k$ is non-causally known at the transmit node, we can apply the same successive DPC encoding as the ZF-DP precoder\cite{CS03} to cancel it. That is, we first encode a first user that suffers no interference from the other users after precoding. Then, the second user is encoded utilizing DPC scheme with regarding the encoded symbols from the first user as known interference. The remaining users are successively encoded in the same manner. For each of the $N$ users, as there are at most $\nu$ users to be considered in the DPC and $\nu\!\ll\!N\!-\!1$, the GZF-DP precoder renders much lower-complexity of the successive DPC operations than the ZF-DP precoder and has similar complexity as the UG-DP precoder with $\nu\!=\!N_g\!-\!1$.

Before deriving the optimal GZF-DP precoder designs, we make some useful notations. Denote the $\nu\!\times\!1$ vectors that comprise the non-zero entries on each column of $\vec{F}$ excluding the main diagonal element as
\bea \label{fn} \vec{f}_n^{\nu}&=&\left[f_{n+1, n}\;, f_{n+2, n}\;,\cdots\;,f_{n\boxplus\nu, n}\right]\rmt.\eea
Moreover, define the $(\nu\!+\!1)\!\times\!(\nu\!+\!1)$ principle sub-matrix $\vec{G}_n^{\nu}$ obtained from $\vec{G}$ as
 {\setlength\arraycolsep{3pt} \bea \label{Gn} \vec{G}_{n}^{\nu}\!=\!\left[\begin{array}{cccc}  g_{n,n}& g_{n,n+1}&\cdots&g_{n,n\boxplus\nu}\\g_{n+1,n}& g_{n+2,n+1}&\cdots&g_{n+1,n\boxplus\nu}\\ \vdots&\vdots&\vdots&\vdots \\g_{n\boxplus\nu,n}&g_{n\boxplus\nu,n+1}&\cdots&g_{n\boxplus\nu,n\boxplus\nu}\end{array} \right]\!\!.  \eea}
\hspace{-1.4mm}and let
\bea \label{gn} \vec{g}_{n}^{\nu}&=&\left[g_{n,n+1}\;,g_{n,n+2}\;,\cdots\;,g_{n,n\boxplus\nu}\right]\rmh.\eea
Then, $\vec{G}_{n+1}^{\nu-1}$ is the $\nu\!\times\!\nu$ principle sub-matrix obtained by further removing the first row and column vectors from $\vec{G}_{n}^{\nu}$.
 
\subsection{Sum-rate Maximization}
We first consider the GZF-DP precoder design for the sum-rate maximization subject to the transmit power constraint (\ref{con2}). The problem can be formulated as
{\setlength\arraycolsep{2pt} \bea \label{prbm3} &&\underset{\vec{F}}{\mathrm{maximize\;\;}}  \sum_{n=1}^{N}\log\left(1+\frac{|f_{n,n}|^2}{N_0}\right) \notag   \\
&&\mathrm{subject\; to\;\;}  \mathrm{Tr}\left(\vec{F}\rmh\vec{G}\vec{F}\right)= P_{\mathrm{T}}.\eea}
\hspace{-1.4mm}Note that, we have changed the power constraint in (\ref{prbm3}) from $\mathrm{Tr}\left(\vec{F}\rmh\vec{G}\vec{F}\right)\!\leq\! P_{\mathrm{T}}$ to $\mathrm{Tr}\left(\vec{F}\rmh\vec{G}\vec{F}\right)\!=\! P_{\mathrm{T}}$. The reason is that, for a solution of (\ref{prbm3}), the equality of the power constraint always holds. This is so, since if $\mathrm{Tr}\left(\vec{F}\rmh\vec{G}\vec{F}\right)\!<\! P_{\mathrm{T}}$ holds, we can scale up $\vec{F}$ to be some $\tilde{\vec{F}}\!=\!\alpha\vec{F}$ ($\alpha\!>\!1$) such that $\mathrm{Tr}\left(\tilde{\vec{F}}\rmh\vec{G}\tilde{\vec{F}}\right)\!=\! P_{\mathrm{T}}$ holds, and with $\tilde{\vec{F}}$ the sum-rate in (\ref{prbm3}) is also increased. By constraining $f_{n,n}\!\geq\!0$, the optimal solution for (\ref{prbm3}) is stated in Theorem 1.
\begin{theorem}
The optimal band-shaped and low-triangular matrix $\vec{F}$ as defined in (\ref{Fmat}) for sum-rate maximization (\ref{prbm3}) satisfies the following conditions
 {\setlength\arraycolsep{2pt}\bea  \label{optfnk} \vec{f}_n^{\nu}=-f_{n,n}\left(\vec{G}_{n+1}^{\nu-1}\right)^{-1}\vec{g}_{n}^{\nu}, \\
\label{wf2} f_{n,n}=\sqrt{N_0\left[\frac{1}{\lambda\hat{g}_n^{\nu}}-1\right]^{+}}, \eea}
\hspace{-1.4mm}where
\bea \label{hatg} \hat{g}_{n}^{\nu}=g_{n,n}-\left(\vec{g}_n^{\nu} \right)\rmh\left(\vec{G}_{n+1}^{\nu-1}\right)^{-1}\vec{g}_{n}^{\nu},\eea
and $\lambda\!>\!0$ is a constant such that the transmit power constraint is satisfied.
\end{theorem}

\begin{proof}
Consider the Lagrangian function
\bea  \label{costfun} \mathcal{L}= \sum_{n=1}^{N}\log\left(1+\frac{|f_{n,n}|^2}{N_0}\right)-\lambda\left(\mathrm{Tr}\left(\vec{F}\rmh\vec{G}\vec{F}\right)- P_{\mathrm{T}}\right),  \eea
where $\lambda$ is the Lagrange multiplier. The necessary conditions \cite{B82} for the optimal solution are
\bea  \label{KKT}   \left.\begin{aligned}
\frac{\partial\mathcal{L} }{\partial f_{n,k}} =0,\; 1\leq n, k&\leq N\;\, \\
 \mathrm{Tr}\left(\vec{F}\rmh\vec{G}\vec{F}\right)- P_{\mathrm{T}}&= 0\;\, \\
  \lambda& \geq 0  \;\, \end{aligned}\right\} .\eea
Note that, with the definitions in (\ref{fn})-(\ref{gn}), the trace term in (\ref{costfun}) can be rewritten as
{\setlength\arraycolsep{2pt} \bea  \label{trace} \mathrm{Tr}\left(\vec{F}\rmh\vec{G}\vec{F}\right)\!=\!\sum_{n=1}^{N}\left[ f_{n,n}\; \left(\vec{f}_n^{\nu}\right)\rmh\right]\! \left[\begin{array}{cc}  g_{n,n}& \left(\vec{g}_n^{\nu}\right)\rmh \\ \vec{g}_n^{\nu}&\vec{G}_{n+1}^{\nu-1} \end{array} \right] \!\left[ \begin{array}{c}f_{n,n}\;\\ \vec{f}_n^{\nu}\end{array}\right]\!. \;     \eea}
\hspace{-1.4mm}Taking the first-order derivatives of $\mathcal{L}$ with respect to $f_{n,n}$ and $\vec{f}_n^{\nu}$, and using (\ref{trace}) results in
{\setlength\arraycolsep{2pt} \bea  \label{der1} \frac{\partial\mathcal{L} }{\partial f_{n,n}} &=&\frac{N_0 f_{n,n}}{
N_0+|f_{n,n}|^2}-\lambda\left(f_{n,n}g_{n,n}+\left(\vec{f}_n^{\nu}\right)\rmh\vec{g}_n^{\nu} \right)\!,\;\\
\label{der2}  \nabla_{\vec{f}_n^{\nu}}\mathcal{L} &=&-\lambda\left( f_{n,n}\left(\vec{g}_n^{\nu}\right)\rmh+\left(\vec{f}_n^{\nu}\right)\rmh\vec{G}_{n+1}^{\nu-1}\right)\rmt \!.  \qquad\eea} 
\hspace{-1.4mm}Then, by setting $\nabla_{\vec{f}_n^{\nu}}\mathcal{L} $ in (\ref{der2}) to zero, the vector $\vec{f}_n^{\nu}$ can be solved for, and the result is given in (\ref{optfnk}).
Inserting (\ref{optfnk}) back into (\ref{der1}) and setting $\partial\mathcal{L} /\partial f_{n,n}$ to zero, we obtain
\bea  \label{wf22} \frac{N_0}{
N_0+|f_{n,n}|^2}=\lambda\left(g_{n,n}- \left(\vec{g}_n^{\nu} \right)\rmh\left(\vec{G}_{n+1}^{\nu-1}\right)^{-1}\vec{g}_{n}^{\nu}\right).\eea
From (\ref{wf22}) it holds that $\lambda\!>\!0$ as $N_0\!>\!0$, since $\hat{g}_{n}^{\nu}\!>\!0$ which will be shown later in Property 1. Using (\ref{hatg}), the optimal $f_{n,n}$ reads
\bea   \label{wf23} \left|f_{n,n}\right|^2=N_0\left[\frac{1}{\lambda\hat{g}_n^{\nu}}-1\right]^{+}. \eea
As we constrain $f_{n,n}$ to be positive, the solution of $f_{n,n}$ is in (\ref{wf2}), which completes the proof.
\end{proof}
With the necessary conditions of $\vec{f}_n^{\nu}$ and $f_{n,n}$ stated in Theorem 1, the constraint in (\ref{prbm3}) can be written as
{\setlength\arraycolsep{2pt} {\bea  \label{trace3} \frac{1}{N_0}\mathrm{Tr}\left(\vec{F}\rmh\vec{G}\vec{F}\right)&=&\frac{1}{N_0}\sum_{n=1}^{N}\hat{g}_n^{\nu}|f_{n,n}|^2  \notag \\&=& \sum_{n=1}^{N}\left[\frac{1}{\lambda}-\hat{g}_n^{\nu}\right]^{+}= \frac{P_{\mathrm{T}}}{N_0}. \eea}
\hspace{-1.4mm}and the sum-rate equals
 \bea \label{dual1}  R^{\mathrm{sum}}=\sum_{n=1}^{N}R_n^{\mathrm{user}},\eea
 where 
\bea R_n^{\mathrm{user}}\!=\!\big[-\log\left( \lambda \hat{g}_n^{\nu}\right)\big]^{+}. \eea
Therefore, to find the optimal solution for (\ref{prbm3}) is equivalent to find an optimal water-level $1/\lambda$ such that (\ref{dual1}) is maximized and (\ref{trace3}) is satisfied, which can be efficiently solved using water-filling scheme\cite{Y06}. Comparing (\ref{wf23}) with (\ref{wf1}), with the GZF-DP precoder a similar water-filling scheme still applies, however, the water-level has changed as $g_{n,n}$ is replaced now by $\hat{g}_n^{\nu}$, due to the preserved inter-user interference.} We state a property below that shows that $\hat{g}_n^{\nu}$ is positive and non-increasing in $\nu$ for all $1\!\leq\!n\!\leq\!N$.

\begin{property}
Under the condition that $\vec{H}$ has full row rank, for $1\!\leq\!n\!\leq\!N$, it holds that 
\bea \label{prop1} 0\!<\! \hat{g}_n^{N-1}\!\leq\!\hat{g}_n^{N-2}\!\leq\!\cdots\!\leq\!\hat{g}_n^1\!\leq\! g_{n,n}.  \eea
\end{property}

\begin{proof}
First we show that for $1\!\leq\!\nu\!\leq\!N\!-\!1$, $0\!<\! \hat{g}_n^{\nu}\!\leq\! g_{n,n}$ holds. Since $\vec{H}$ has full row rank, $\vec{G}\!\succ\!\vec{0}$. Consequently, $\vec{G}_{n+1}^{\nu-1}$ and $\vec{G}_{n}^{\nu}$ are also positive-definite as principle sub-matrices of $\vec{G}$. Hence, $\left(\vec{g}_n^{\nu} \right)\rmh\left(\vec{G}_{n+1}^{\nu-1}\right)^{-1}\vec{g}_{n}^{\nu}\!\geq\!0$, and $\hat{g}_n^{\nu}\!\leq\!g_{n,n}$ follows from (\ref{hatg}). On the other hand, from the definition, $\vec{G}_{n}^{\nu}$ equals
\vspace{-4mm}
\bea \label{Gn1} \vec{G}_{n}^{\nu}=\left[\!\begin{array}{cc} g_{n,n}&\left(\vec{g}_n^{\nu} \right)\rmh\\ \vec{g}_n^{\nu} & \vec{G}_{n+1}^{\nu-1} \end{array}\!\right]\!\!. \eea
Hence, $\hat{g}_n^{\nu}$ is the Schur-complement\cite{Z05} of $g_{n,n}$, and by utilizing the matrix-inversion lemma \cite{GL13}, the inverse $\left(\vec{G}_{n}^{\nu}\right)^{-1}\!\succ\!\vec{0}$ is in (\ref{inv1}), which shows that $\hat{g}_n^{\nu}\!>\!0$.

Next we show that, $\hat{g}_n^{\nu}\!\leq\! \hat{g}_n^{\nu-1}$ holds for $1\!\leq\!n\!\leq\!N$. Firstly, for $n\!>\!N\!-\!\nu$, by definitions (\ref{gn}) and (\ref{Gn}), the equalities $\vec{g}_n^{\nu}\!=\!\vec{g}_n^{\nu-1}\!$ and $\vec{G}_{n+1}^{\nu-1}\!=\!\vec{G}_{n+1}^{\nu-2}$ hold. Hence, from (\ref{hatg}), $\hat{g}_n^{\nu}\!=\! \hat{g}_n^{\nu-1}$ holds. Secondly, for $1\!\leq\!n\!\leq\!N\!-\!\nu$, $\vec{G}_{n}$ in (\ref{Gn1}) can also be rewritten as
\bea  \vec{G}_{n}^{\nu}=\left[\!\begin{array}{cc} \vec{G}_{n}^{\nu-1}&\left(\tilde{\vec{g}}_n^{\nu}\right)\rmh \\ \tilde{\vec{g}}_{n}^{\nu}& g_{n+\nu, n+\nu} \end{array}\!\right]\!\!,  \eea
where $\tilde{\vec{g}}_n^{\nu}\!=\!\left[g_{n+\nu,n},\, g_{n+\nu,n+1},\, \cdots, \,g_{n+\nu,n+\nu-1}\right]$. By utilizing the matrix-inversion lemma again, the inverse $\left(\vec{G}_{n}^{\nu}\right)^{-1}$ can also be written in (\ref{inv2}). From (\ref{inv1}) we know that, $\left(\hat{g}_n^{\nu-1}\right)^{-1}$ is the first diagonal element of $\left(\vec{G}_{n}^{\nu-1}\right)^{-1}$, while $\left(\hat{g}_n^{\nu}\right)^{-1}$ is the first diagonal element of $\left(\vec{G}_{n}^{\nu}\right)^{-1}$ and hence, the first diagonal element of $\left(\vec{G}_{n}^{\nu-1}\!-\!\frac{\left(\tilde{\vec{g}}_n^{\nu}\right)\rmh\tilde{\vec{g}}_n^{\nu}}{g_{n+\nu, n+\nu}}\right)^{-1}$ from (\ref{inv2}). Using the Woodbury matrix identity\cite{GL13}, $ \left(\vec{G}_{n}^{\nu-1}\!-\!\frac{\left(\tilde{\vec{g}}_n^{\nu}\right)\rmh\tilde{\vec{g}}_n^{\nu}}{g_{n+\nu, n+\nu}}\right)^{-1}\!\succeq\!\left(\vec{G}_{n}^{\nu-1}\right)^{-1}$ holds. Therefore, $\left(\hat{g}_n^{\nu}\right)^{-1}\!\geq\!\left(\hat{g}_n^{\nu-1}\right)^{-1}$ holds, and $\hat{g}_n^{\nu}\!\leq\! \hat{g}_n^{\nu-1}$ follows, which completes the proof.
\end{proof}

\begin{figure*}[b]
\vspace{-8mm}
\hrulefill
\vspace{2mm}
{\setlength\arraycolsep{4pt} 
 \bea \label{inv1} \left(\vec{G}_{n}^{\nu}\right)^{-1} \!=\!\left[\!\begin{array}{cc} \left(\hat{g}_n^{\nu}\right)^{-1}&-\left(\hat{g}_n^{\nu}\right)^{-1}\left(\vec{g}_n^{\nu} \right)\rmh\left(\vec{G}_{n+1}^{\nu-1}\right)^{-1}\\ -\left(\hat{g}_n^{\nu}\right)^{-1}\left(\vec{G}_{n+1}^{\nu-1}\right)^{\!-1}\!\vec{g}_n^{\nu}& \left(\vec{G}_{n+1}^{\nu-1}\right)^{\!-1}\!+\left(\hat{g}_n^{\nu}\right)^{\!-1}\left(\vec{G}_{n+1}^{\nu-1}\right)^{-1}\vec{g}_n^{\nu}\left(\vec{g}_n^{\nu} \right)\rmh\left(\vec{G}_{n+1}^{\nu-1}\right)^{\!-1} \end{array}\!\!\right]\!\!. \eea}
 \vspace{2mm}
\hspace{-1.4mm}\hrulefill
\vspace{1mm}
{\setlength\arraycolsep{4pt} 
 \bea \label{inv2} \left(\vec{G}_{n}^{\nu}\right)^{-1} \!\!=\!\!\left[\!\begin{array}{cc} \left(\vec{G}_{n}^{\nu-1}\!-\!\frac{\left(\tilde{\vec{g}}_n^{\nu}\right)\rmh\tilde{\vec{g}}_n^{\nu}}{g_{n+\nu, n+\nu}}\right)^{\!-1}&-\left(\vec{G}_{n}^{\nu-1}\!-\!\frac{\left(\tilde{\vec{g}}_n^{\nu}\right)\rmh\tilde{\vec{g}}_n^{\nu}}{g_{n+\nu, n+\nu}}\right)^{-1}\frac{\left(\tilde{\vec{g}}_n^{\nu}\right)\rmh}{g_{n+\nu, n+\nu}}\\ -\frac{-\tilde{\vec{g}}_n^{\nu}}{g_{n+\nu, n+\nu}}\left(\vec{G}_{n}^{\nu-1}\!-\!\frac{\left(\tilde{\vec{g}}_n^{\nu}\right)\rmh\tilde{\vec{g}}_n^{\nu}}{g_{n+\nu, n+\nu}}\right)^{\!-1}& \frac{1}{g_{n+\nu, n+\nu}}\!+\!\frac{\tilde{\vec{g}}_n^{\nu}}{g_{n+\nu, n+\nu}^2}\left(\vec{G}_{n}^{\nu-1}\!-\!\frac{\left(\tilde{\vec{g}}_n^{\nu}\right)\rmh\tilde{\vec{g}}_n^{\nu}}{g_{n+\nu, n+\nu}}\right)^{-1}\left(\tilde{\vec{g}}_n^{\nu}\right)\rmh
 
   \end{array}\!\!\right]\!\!.\quad\; \; \eea}
   \vspace{-9mm}
 \end{figure*}
 
As $\hat{g}_n^{\nu}\!\leq\! g_{n,n}$, from (\ref{trace3}) in general the water-level $1/\lambda$ is actually non-increasing when $\nu$ increases. Therefore, not all the user-rates are increased with a larger $\nu$. For instance, for the last user, as $\hat{g}_N^{\nu}\!=\! g_{N,N}$ for all $\nu$, the user-rate $R_{N}^{\mathrm{user}}$ is non-increasing as $\nu$ increases. In general, we have the following corollary.

\begin{corollary}
If $\nu$ is increased from $\nu_1$ to $\nu_1\!+\!1$ for the GZF-DP precoder, as for $n\!\geq\!N\!-\!\nu_1$, $\hat{g}_n^{\nu}\!=\!\hat{g}_n^{\nu-1}\!$ holds, and as a result of the non-increasing water-level, the user-rates of the last $\nu_1\!+\!1$ users are also non-increasing. 
\end{corollary}
However, the sum-rate never decrease with a larger $\nu$, which is stated in the below property.

\begin{property}
If $\nu_2\!>\!\nu_1$, the sum-rate $R^{\mathrm{sum}}$ obtained with the GZF-DP precoder with $\nu\!=\!\nu_2$ is no less than that obtained with $\nu\!=\!\nu_1$. However, under the case that the channel $\vec{H}$ itself is band-shaped with only the elements along the main diagonal and the first $\nu_1$ lower-diagonals can take non-zero values, increasing $\nu$ to be larger than $\nu_1$ will not further increase $R^{\mathrm{sum}}$.
\end{property}
\begin{proof}
The first statement holds from the fact that the effective channel $\vec{F}$ with $\nu\!=\!\nu_1$ is a subset of $\vec{F}$ with $\nu\!=\!\nu_2$. Next we prove the second statement by showing that $\hat{g}_n^{\nu}\!=\!\hat{g}_n^{\nu_1}\!$ for any $n$ and $\nu\!>\!\nu_1$, under the condition that $\vec{H}$ is band-shaped with only the elements along the main diagonal and the first $\nu_1$ lower-diagonals can take non-zero values. Therefore, in such a case, the sum-rate $R^{\mathrm{sum}}$ obtained with $\nu\!>\!\nu_1$ is equal to $R^{\mathrm{sum}}$ with $\nu\!=\!\nu_1$.

We first show that, for $n\!=\!1$, $\hat{g}_1^{\nu}\!=\!\hat{g}_1^{\nu_1}\!$ holds for $\nu\!>\!\nu_1$. We decompose $\vec{G}$ and $\vec{H}\vec{H}\rmh$ into block forms as
 \bea \vec{G}=\left[\!\begin{array}{cc} \vec{G}_{1}^{\nu}&\vec{G}_2\rmh\\ \vec{G}_2& \vec{G}_3 \end{array}\!\right]\!\!,  \;\;
\vec{H}\vec{H}\rmh=\left[\!\begin{array}{cc} \vec{B}_{1}&\vec{B}_2\rmh\\ \vec{B}_2& \vec{B}_3 \end{array}\!\right]\!\!, \!\!\!\!\!\! \eea
where sub-matrix $\vec{G}_{1}^{\nu}$ follows the definition in (\ref{Gn}) and sub-matrices $\vec{G}_2$, $\vec{G}_3$ are deduced from $\vec{G}_{1}^{\nu}$. Similarly, sub-matrix $\vec{B}_{1}$ has the same size as $\vec{G}_{1}^{\nu}$, and sub-matrices $\vec{B}_2$, $\vec{B}_3$ are deduced from $\vec{B}_{1}$. As $\vec{G}\!=\!\left(\vec{H}\vec{H}\rmh\right)^{-1}$, following the matrix inversion lemma we have
\bea \label{prop21} \left(\vec{G}_{1}^{\nu}\right)^{-1}=\vec{B}_1-\vec{B}_2\rmh\vec{B}_3^{-1}\vec{B}_2. \eea
As $\vec{H}$ is band-shaped, when $\nu\!\geq\!\nu_1$, the first row vector in $\vec{B}_2\rmh$ comprises all zero elements. Consequently, from (\ref{prop21}) the first diagonal element of $\left(\vec{G}_{1}^{\nu}\right)^{-1}$, which is $\left(\hat{g}_{1}^{\nu}\right)^{-1}$, is equal to the first diagonal element of $\vec{B}_1$. Hence, we have
\bea \label{prop22} \hat{g}_{1}^{\nu}=|h_1(1)|^{-2},\;\; \nu\!\geq\!\nu_1, \eea
where $h_1(1)$ is the first tap of the channel vector corresponding to the first user, 

For $n\!>\!1$, we can permute the principle sub-matrix $\vec{G}_{n}^{\nu}$ to the upper-left corner with a permutation matrix $\vec{Q}$ such that,
\bea \vec{Q}\vec{G}\vec{Q}\rmh=\left[\!\begin{array}{cc} \vec{G}_{n}^{\nu}&\tilde{\vec{G}}_2\rmh\\ \tilde{\vec{G}}_2& \tilde{\vec{G}}_3 \end{array}\!\right]\!\!,  \eea
where $\tilde{\vec{G}}_2$, $\tilde{\vec{G}}_3$ are deduced from $\vec{G}_{n}^{\nu}$. We also permute $\vec{H}\vec{H}\rmh$ accordingly such that
\bea \vec{Q}\vec{H}\vec{H}\rmh\vec{Q}\rmh=\left[\!\begin{array}{cc} \tilde{\vec{B}}_1&\tilde{\vec{B}}_2\rmh\\ \tilde{\vec{B}}_2& \tilde{\vec{B}}_3 \end{array}\!\right]\!\!,  \eea
where sub-matrices $\tilde{\vec{B}}_{1}, \tilde{\vec{B}}_2, \tilde{\vec{B}}_3$ are defined similarly as before. As $\vec{Q}\vec{G}\vec{Q}\rmh\!=\!\left(\vec{Q}\vec{H}\vec{H}\rmh\vec{Q}\rmh\right)^{-1}$ holds, following (\ref{prop21}) and (\ref{prop22}) we have
\bea \label{prop23} \hat{g}_{n}^{\nu}=|h_n(n)|^{-2},\;\; \nu\!\geq\!\nu_1, \eea
where $h_n(n)$ is the $n$th tap of the channel vector corresponding to the $n$th user, which is transfered to be the first user after permutation. Therefore, with $\hat{g}_n^{\nu}$ given in (\ref{prop23}), it holds that, $\hat{g}_n^{\nu}\!=\!\hat{g}_n^{\nu_1}\!$ for any $n$ and $\nu\!>\!\nu_1$, which completes the proof.
\end{proof}

Property 2 reveals that if $\vec{H}$ is banded, further increasing the band-size of $\vec{F}$ to be larger than the band-size of $\vec{H}$ will not increase the sum-rate. Moreover, for a band-shaped $\vec{H}$, $\hat{g}_{n}^{\nu}$ can be easily calculated through (\ref{prop23}) for $\nu\!\geq\!\nu_1$. Next, we show that the GZF-DP precoder design actually provides a unified framework of the previous ZF based precoder designs.
\begin{corollary}
With $\nu\!=\!0$, the GZF-DP precoder becomes the linear ZF precoder without DPC; while with $\nu\!=\!N\!-\!1$, the GZF-DP precoder is identical to the ZF-DP precoder. In addition, the UG-DP precoder is inferior to the GZF-DP precoder with $\nu\!=\!N_g\!-\!1$.
\end{corollary}
\begin{proof}
When $\nu\!=\!0$, $\hat{g}_n^{\nu}\!=\!g_{n,n}$ for all $n$, and the GZF-DP precoder is thusly identical to the linear ZF precoder. On the other hand, when $\nu\!=\!N\!-\!1$, the maximization (\ref{prbm3}) can be formulated as the same problem in (\ref{prbm2}), which shows the trade-off between the sum-rate and the complexity of successive DPC. Moreover, as the UG-DP can be reviewed as a special case of the GZF-DP with $\nu\!=\!N_g\!-\!1$, the UG-DP precoder is inferior to the GZF-DP precoder in general.
\end{proof}
Although the ZF-DP precoder provides the highest sum-rate, as shown in Corollary 1, it sacrifices user-rates of some of the last users. As a generalization of the ZF-DP precoder, the GZF-DP precoder, however, can provide a trade-off between the sum-rate increment and the user-rate decrement through the parameter $\nu$. Below we illustrate with an example to show different designs of the linear ZF precoder, the UG-DP precoder, and the proposed GZF-DP precoder.
 
 \begin{figure*}[!b]
{\small
\vspace{-8mm}
\hrulefill
\vspace{2mm}
{\setlength\arraycolsep{1.5pt} 
\bea \vec{F}_{\mathrm{ZF}}\!=\!\left[\!\begin{array}{cccc}  4.376 &~&~&~\\ ~&5.238&~&~ \\ ~&~&4.436&~\\ ~&~&~& 4.407\end{array} \!\right]\!\!, \;
\vec{F}_{\mathrm{UG-DP}}\!=\!\left[\!\begin{array}{cccc}    4.899&0&~&~\\   -1.140 \!+\! 2.340i & -5.217&~&~ \\ ~&~&4.490 &0\\ ~&~&0.489 \!+\! 0.607i&4.389 \end{array} \!\right]\!\!, \notag\eea}
\vspace{2mm}
\hspace{-1.4mm}\hrulefill
\vspace{2mm}
{\setlength\arraycolsep{2pt} 
\bea
 \vec{F}_{\mathrm{GZF-DP},\;\nu=1}\!=\!\left[\!\begin{array}{cccc}    4.910&0&~&~\\     -1.143 \!+\! 2.345i & 5.784&~&~ \\ ~& 2.034 \!+\! 0.416i &4.501 &0\\ ~&~&0.490 \!+\! 0.609i&4.400 \end{array} \!\right]\!\!. \notag\eea}
   \vspace{-12mm}
   }
 \end{figure*}
 
\begin{example}
Assume $N_0\!=\!1$, $P_{\mathrm{T}}\!=\!10$ dB, and consider an MISO channel with 4 transmit antennas and 4 single-antenna users as ($i=\sqrt{-1}$) 
{\setlength\arraycolsep{4pt} 
\bea \vec{H}=\left[\!\begin{array}{cccc} 1+4i&4+3i&2+3i&3+3i\\ 4+1i&1+4i&1+1i&2+4i \\ 2+3i&1+4i&3+3i&4+3i\\ 4+4i&2+3i&1+4i&2+2i\end{array} \!\right]\!. \notag\eea} 
\hspace{-2mm}The sum-rates (bits/channel use) of the ZF precoder, the UG-DP precoder with $N_{\mathrm{g}}\!=\!2$, and the GZF-DP precoder with $\nu\!=\!1$ are equal to
\bea  R_{\mathrm{ZF}}^{\mathrm{sum}}\!=\!17.885,\;\,R_{\mathrm{UG-DP}}^{\mathrm{sum}}\!=\!18.206, \;\,
\mathrm{and}\;\, R_{\mathrm{GZF-DP,\;\nu=1}}^{\mathrm{sum}}\!=\!18.514,\notag \eea
respectively. The optimal effective channels $\vec{F}$ are listed at the bottom of this page.
\end{example}
With Example 1, the user-rates corresponding to different precoders are equal to
{\setlength\arraycolsep{2pt} \bea  R_{\mathrm{ZF}}^{\mathrm{user}}&=&[ 4.333, \,4.830, \,4.370, \,4.352], \notag \\
 R_{\mathrm{GZF-DP},\; \nu=1}^{\mathrm{user}}&=&[ 4.650, \,5.106, \,4.410, \,4.348], \notag \\
 R_{\mathrm{GZF-DP}, \;\nu=2}^{\mathrm{user}}&=&[ 5.394, \,6.047, \,4.387, \,4.324].\notag \eea}
\hspace{-1.4mm}As it can been seen, although the sum-rate is increased from $\nu\!=\!0$ to 1, the user-rate of the last user is decreased. Further, from $\nu\!=\!1$ to 2, the user-rates of the last two users are also decreased, which are aligned with Corollary 1. Especially for the last user, the user-rate is continuously decreasing when $\nu$ increases from 0 to 2. Therefore, instead of maximizing the sum-rate, it is also meaningful to consider maximizations of user-rate, which is usually used as a measurement for the fairness of the QoS. 

Next we discuss the minimum user-rate maximization with the proposed GZF-DP precoder.

\subsection{Minimum User-rate Maximization}
For minimum user-rate maximization, the design of the GZF-DP precoder is formulated as
{\setlength\arraycolsep{2pt} \bea \label{prbm41} &&\underset{\vec{F},\, R^{\mathrm{user}}}{\mathrm{maximize\;\;}}  R^{\mathrm{user}} \notag \\
&&\mathrm{\,subject\; to\;\,}   R^{\mathrm{user}}\leq\log\left(1+\frac{|f_{n,n}|^2}{N_0}\right), 1\leq n\leq N \notag   \\
&&\qquad \qquad\quad\!\mathrm{Tr}\left(\vec{F}\rmh\vec{G}\vec{F}\right)\leq P_{\mathrm{T}},\eea}
\hspace{-1.4mm}where the matrices $\vec{F}$ and $\vec{G}$ are the same as defined for the sum-rate maximization and we constrain $f_{n,n}\!\geq\!0$. Following similar arguments as for the sum-rate maximization, it also holds that the equality in the power constraint always holds for an optimal solution of (\ref{prbm41}). Furthermore, we have the below lemma.

\begin{lemma}
For an optimal solution $\vec{F}$ of (\ref{prbm41}), it holds that $R^{\mathrm{user}}\!=\!\log\left(1+\frac{|f_{n,n}|^2}{N_0}\right)$ for all $n$.
 \end{lemma}
\begin{proof}
For an optimal solution $\vec{F}$ of (\ref{prbm41}), we denote the maximal and minimal user-rates as $ R_{n_1}^{\mathrm{user}}$ and $R_{n_2}^{\mathrm{user}}$, respectively, which equal
\bea \label{Rni} R_{n_i}^{\mathrm{user}}=\log\left(1+\frac{|f_{n_i,n_i}|^2}{N_0}\right), \quad i=1, 2.  \eea
Then, the minimum user-rate is equal to $R_{n_2}^{\mathrm{user}}$. We further denote the transmit powers of user $n_1$ and $n_2$ as $P_1$ and $P_2$, respectively. According to (\ref{trace}), it holds that
\bea \label{pi} P_i= g_{n_i,n_i}|f_{n_i,n_i}|^2+2\mathcal{R}\left\{\left(\vec{f}_{n_i}^{\nu}\right)\rmh\vec{g}_{n_i}^{\nu}f_{n_i,n_i}\right\}+\left(\vec{f}_{n_i}^{\nu}\right)\rmh\vec{G}_{n_i+1}^{\nu-1}\vec{f}_{n_i}^{\nu}, \quad i=1, 2.  \eea
Now let\rq{}s assume $R_{n_1}\!>\! R_{n_2}$, that is, the maximal user-rate is strictly larger than the minimal user-rate. Then, we can scale $f_{n_i,n_i}$ and $\vec{f}_{n_i}^{\nu}$ to be $\tilde{f}_{n_i,n_i}\!=\!\alpha_i f_{n_i,n_i}$ and $\tilde{\vec{f}}_{n_i}^{\nu}\!=\!\alpha_i\vec{f}_{n_i}^{\nu}$, respectively, where $\alpha_2\!>\!1\!>\!\alpha_1$ and
\bea \alpha_1=\sqrt{1+\frac{(1-\alpha_2^2)P_2}{P_1}}. \notag \eea
Note that, according to (\ref{pi}), with such a scaling operation, the total transmit power of user $n_1$ and $n_2$ remains the same, that is, $\alpha_1^2P_1\!+\!\alpha_2^2P_2\!=\!P_1\!+\!P_2$. However, according to (\ref{Rni}), the user-rate with such a scaling increases $R_{n_2}^{\mathrm{user}}$ and decreased $R_{n_1}^{\mathrm{user}}$. Hence, the minimum user-rate can therefore be increased, which contradicts to the assumption that $\vec{F}$ is optimal. Therefore, for an optimal $\vec{F}$, $R_{n_1}^{\mathrm{user}}\!=\! R_{n_2}^{\mathrm{user}}$ holds, which shows that all user-rates are equal to each other for an optimal $\vec{F}$ of (\ref{prbm41}).
\end{proof}
 
With the above arguments, we can change (\ref{prbm41}) to the equivalent problem 
{\setlength\arraycolsep{2pt} \bea \label{prbm4} &&\underset{\vec{F},\, R^{\mathrm{user}}}{\mathrm{maximize\;\;}}  R^{\mathrm{user}} \notag \\
&&\mathrm{\,subject\; to\;\,}   R^{\mathrm{user}}=\log\left(1+\frac{|f_{n,n}|^2}{N_0}\right), 1\leq n\leq N \notag   \\
&&\qquad \qquad\quad\!\mathrm{Tr}\left(\vec{F}\rmh\vec{G}\vec{F}\right)= P_{\mathrm{T}}.\eea}
\hspace{-1.4mm}Then, the necessary conditions for an optimal solution $\vec{F}$ is stated in Theorem 2.
\begin{theorem}{
The optimal band-shaped and low-triangular matrix $\vec{F}$ in (\ref{Fmat}) for user-rate maximization (\ref{prbm4}) shall satisfy the the conditions that}, the optimal $\vec{f}_n^{\nu}$ is in (\ref{optfnk}) and $f_{n,n}$ equals
\bea \label{wf3} f_{n,n}=\sqrt{N_0\left[\frac{1}{\lambda_n\hat{g}_n^{\nu}}-1\right]^{+}}, \eea
where $\lambda_n\!>\!0$ are a set of constants such that the transmit power constraint is satisfied.\end{theorem} 
\begin{proof}
The Lagrangian function for multiple constraints in this case reads
 \bea  \label{costfun2} \mathcal{L}= R^{\mathrm{user}}-\sum_{n=1}^{N}\mu_n\left(R^{\mathrm{user}}-\log\left(1+\frac{|f_{n,n}|^2}{N_0}\right)\right) -\lambda\left(\mathrm{Tr}\left(\vec{F}\rmh\vec{G}\vec{F}\right)- P_{\mathrm{T}}\right),  \eea
and the necessary conditions are
\bea   \label{KKT2}   \left.\begin{aligned}
\frac{\partial \mathcal{L} }{\partial f_{n,k}}\! =\!0,\; 1\leq n, k&\leq N\;\, \\
\sum_{n=1}^N\mu_n\!=\!1,\;\mu_n\geq 0, \; 1\leq n&\leq N\;\, \\
R^{\mathrm{user}}=\log\left(1+\frac{|f_{n,n}|^2}{N_0}\right), \; 1\leq n&\leq N \;\, \\
  \mathrm{Tr}\left(\vec{F}\rmh\vec{G}\vec{F}\right)- P_{\mathrm{T}}&= 0\;\, \\
    \lambda&\geq 0 \;\, \end{aligned}\right\}. 
\eea
The first-order derivatives of $\mathcal{L}$ with respect to $f_{n,n}$ is
 \bea  \label{der12} \frac{\partial\mathcal{L} }{\partial f_{n,n}} =\frac{\mu_n N_0f_{n,n}}{
N_0+|f_{n,n}|^2}-\lambda\left(f_{n,n}g_{n,n}+\left(\vec{f}_n^{\nu}\right)\rmh\vec{g}_n^{\nu} \right), \eea
while the gradient of $\mathcal{L}$ with respect to $\vec{f}_n^{\nu}$ is in (\ref{der2}). Then, from (\ref{der2}) the optimal $\vec{f}_n^{\nu}$ is solved in (\ref{optfnk}), and by inserting (\ref{optfnk}) back into (\ref{der12}) and setting the derivative to zero, we obtain
\bea  \label{optder1} \frac{N_0\mu_n}{
N_0+|f_{n,n}|^2}=\lambda\left(g_{n,n}-\left(\vec{g}_{n}^{\nu}\right)\rmh\left(\vec{G}_{n+1}^{\nu-1}\right)^{-1}\vec{g}_{n}^{\nu}\right).\eea
Hence, as $N_0\!>\!0$, from (\ref{optder1}) it holds that $\lambda\!>\!0$ and $\mu_n\!>\!0$ for all $n$. Otherwise, if either $\lambda\!=\!0$ or $\mu_n\!=\!0$ for some $n$, from (\ref{optder1}) it holds that $\lambda\!=\!\mu_n\!=\!0$ for all $n$, which contradicts the second necessary condition in (\ref{KKT2}) (due to $\partial\mathcal{L} /\partial R^{\mathrm{user}}\!=\!0$). By setting $\lambda_n\!=\!\lambda/\mu_n\!>\!0$ and from (\ref{optder1}) the optimal $f_{n,n}$ equals
\bea  \left|f_{n,n}\right|^2=N_0\left[\frac{1}{\lambda_n\hat{g}_n^{\nu}}-1\right]^{+}, \notag\eea
where $\hat{g}_{n,n}$ is defined in (\ref{hatg}), and the optimal $f_{n,n}$ is then in (\ref{wf3}).  
\end{proof}

With the necessary conditions of an optimal $\vec{F}$ in Theorem 2, the user-rate is equal to the minimum user-rate for all users, that is,
\bea  \label{raten} R^{\mathrm{user}}=\big[-\log\left( \lambda_n \hat{g}_n^{\nu}\right)\big]^{+},\quad 1\leq n\leq N,\eea
and the power constraint can be written as
\bea  \label{trace4}\sum_{n=1}^{N}\left[\frac{1}{\lambda_n}-\hat{g}_n^{\nu}\right]^{+}= \frac{P_{\mathrm{T}}}{N_0}. \eea
Note that, different from the sum-rate maximization, now the water-level $1/\lambda_n$ varies for different users. From (\ref{raten}) and (\ref{trace4}), the minimum user-rate can be solved for in closed-form,
\bea  R^{\mathrm{user}}=\log\!\left(1+\frac{P_{\mathrm{T}}}{N_0\sum\limits_{n=1}^{N}\hat{g}_n^{\nu}}\right),\eea
and the optimal $f_{n,n}$ equals
\bea  f_{n,n}=\sqrt{N_0\left(2^{R^{\mathrm{user}}}-1\right)},\quad 1\leq n\leq N. \eea
Although with the sum-rate maximization some user-rates may be decreased with a larger $\nu$ as shown in Corollary 1, for minimum user-rate maximization, $R^{\mathrm{user}}$ will not be decreased by a larger $\nu$. Further, as the maximal minimum user-rate $R^{\mathrm{user}}$ in (\ref{raten}) is uniquely determined by the values of $\hat{g}_n^{\nu}$, we have the below property.
\begin{property}
The conclusions drawn for sum-rate $R^{\mathrm{sum}}$ in Property 2 also stand for minimum user-rate $R^{\mathrm{user}}$.
\end{property}

\subsection{Optimal User-Orderings}
By permuting the order of the $N$ users with an $N\!\times\!N$ permutation matrix $\vec{Q}$, the received signal model (\ref{md1}) reads
\bea \label{md4} \vec{Q}\vec{y} =\vec{Q}\vec{H}\vec{P}\vec{x}+\vec{Q}\vec{z}.\eea
Changing the order of the users may impact\footnote{This is true for cases $0\!<\!\nu\!<\!N$. For $\nu\!=\!0$, i.e., the linear ZF precoder, as the inter-user interferences are completely nulled out, different user-orderings have no impact on both the sum-rate or minimum user-rate maximizations.} the optimizations in (\ref{prbm3}) and (\ref{prbm4}), due to that the matrix $\vec{G}$ is updated with $\tilde{\vec{G}}\!=\!\vec{Q}\vec{G}\vec{Q}\rmh$ and the power constraint changes to,
\bea  \label{con3} \mathrm{Tr}\left(\vec{F}\rmh\tilde{\vec{G}}\vec{F}\right)\leq P_{\mathrm{T}},\eea
Denoting the set that comprises all possible user-orderings as $\mathcal{P}$, and as the size $|\mathcal{P}|\!=\!N!$, it is infeasible to find an optimal ordering in a brute-force manner for large values of $N$. Therefore, we next introduce two efficient suboptimal user-ordering algorithms for the sum-rate and the minimum user-rate maximizations for $0\!<\!\nu\!<\!N$ that have complexity orders $\mathcal{O}\left({N}\choose{\nu\!+\!1}\right)$ and $\mathcal{O}(N)$, respectively. We start with the user-ordering for the sum-rate maximization (\ref{prbm3}). From (\ref{dual1}), the optimal user-ordering $\mathcal{U}\!\in\!\mathcal{P}$ shall minimize the product\footnote{Without loss of generality, we assume $\lambda \hat{g}_{n,n}\!\geq\!1$ holds for all users for both sum-rate and minimum user-rate maximizations.}, 
 \bea \label {p} \mathcal{U}^{\mathrm{opt}}=\argmin_{\mathcal{U}\in\mathcal{P}}\lambda ^N\prod_{n=1}^{N}\hat{g}_{n}^{\nu}.   \eea
Denoting $q\!=\!\prod\limits_{n=1}^{N}\hat{g}_{n}^{\nu}$ and since
\bea \lambda=\left(\frac{ P_{\mathrm{T}}}{N_0}+\sum_{n=1}^{N}\hat{g}_{n}^{\nu}\right)^{-1}\leq\left(\frac{ P_{\mathrm{T}}}{N_0}+N q^{\frac{1}{N}}\right)^{-1}, \notag \eea
it holds that
\bea \label{p1} \lambda ^N q\leq\left(\frac{ P_{\mathrm{T}}}{N_0 q^{\frac{1}{N}}}+N\right)^{-N}. \eea \notag
Instead of directly minimizing (\ref{p}), from (\ref{p1}) we can minimize the product $q$ instead. On the other hand, from (\ref{Gn1}) and utilizing the matrix determinant lemma \cite{D08}, $\hat{g}_{n}^{\nu}$ can be rewritten as
$ \hat{g}_{n}^{\nu}\!=\!\det{\vec{G}_{n}^{\nu}}/\det{\vec{G}_{n+1}^{\nu-1}}$, and $q$ equals
\bea \label{prod2} q=\prod_{n=1}^{N}\frac{\det{\vec{G}_{n}^{\nu}}}{\det{\vec{G}_{n+1}^{\nu-1}}} . \eea
By noticing that the sub-matrix $\vec{G}_{n}^{\nu}$ comprises $\vec{G}_{n+1}^{\nu-1}$ and an extra row and column vectors corresponding to the $n$th user, we can recursively order the users according to (\ref{prod2}) as follows.

At a first stage, to minimize $\hat{g}_{1}^{\nu}$ we first find the best $\nu\!+\!1$ users that minimize $\det{\vec{G}_{1}^{\nu}}$, which needs to search over in total ${N}\choose{\nu\!+\!1}$ possible user combinations\footnote{Note that, the ordering of the $\nu\!+\!1$ users inside each combination is independent with $\det{\vec{G}_{n}^{\nu}}$ since the determinant is invariant under the operation that permutes the row and column vectors in the same manner.}. We denote the index set of the obtained $\nu\!+\!1$ users as $\mathcal{J}_1$. Then, in a second step, we select one single user from the chosen $\nu\!+\!1$ users that maximize $\det{\vec{G}_2^{\nu-1}}$, where $\det{\vec{G}_2^{\nu-1}}$ is obtained by removing the corresponding row and column vectors of the selected user in $\vec{G}_{1}^{\nu}$. One such user is selected to be the first user and set $\mathcal{U}(1)$ to its user-index.

At a second stage, we continue to order the remaining $N\!-\!1$ users, with $\nu$ users within the index set $\mathcal{J}_2\!=\!\mathcal{J}_1\!\setminus \mathcal{U}(1)$. In order to minimize $\hat{g}_{2}^{\nu}$, we first add another user from the remaining $N\!-\!\nu\!-\!1$ users to the $\nu$ users in $\mathcal{J}_2$ and calculate $\det\vec{G}_{2}^{\nu}$ corresponding to the selected $\nu\!+\!1$ users. The user from the remaining $N\!-\!\nu\!-\!1$ users that minimize $\det{\vec{G}_{2}^{\nu}}$ is selected, which needs $N\!-\!\nu\!-\!1$ operations. We update $\mathcal{J}_1$ as $\mathcal{J}_2$ plus the selected user-index. Then, we repeat the second step at the first stage to select one user from $\mathcal{J}_2$ (not $\mathcal{J}_1$ in order to keep the value of  $\hat{g}_{1}^{\nu}$ unchanged) to maximize $\det{\vec{G}_{3}^{\nu-1}}$, and set $\mathcal{U}(2)$ to the index of that user.

\vspace{-2mm}
\begin{algorithm}[ht!]
	\caption{User-ordering for sum-rate maximization with the GZF-DP precoder.}
	\label{alg:2}
      \begin{algorithmic}[1]
      \vspace{1mm}
      \STATE Initialize $n\!=\!1$ and $\mathcal{I}_1\!=\!\mathcal{I}_2\!=\![1,2,\cdots,N]$. \\
      \STATE Search over all ${N}\choose{\nu\!+\!1}$ possible combinations to find the best $\nu\!+\!1$ users that minimizes the determinant of the principle sub-matrix $\det{\vec{G}_{1}^{\nu}}$ introduced by their indexes, and denote the best user-combination as $\mathcal{J}_1$, then set $\mathcal{J}_2\!=\!\mathcal{J}_1$.\\
       \STATE Select one single user from all users in $\mathcal{J}_2$ to maximize $\det{\vec{G}_{2}^{\nu-1}}$, and denote its user-index as $\mathcal{U}(n)$.\\
       \STATE Update $\mathcal{I}_1\!=\!\mathcal{I}_1\!\setminus\!\mathcal{U}(n)$, $\mathcal{J}_2\!=\!\mathcal{J}_1\!\setminus\!\mathcal{U}(n)$, $\mathcal{I}_2\!=\!\mathcal{I}_1\!\setminus\!\mathcal{J}_2$, and set $n\!=\!n\!+\!1$.\\
       \STATE Replace the index $\mathcal{U}(n\!-\!1)$ in $\mathcal{J}_1$ with another user-index from the $N\!-\!\nu\!-\!n$ users in $\mathcal{I}_2$, such that $\det{\vec{G}_{n}^{\nu}}$ introduced by the updated $\mathcal{J}_1$ is minimized, and keep the updated $\mathcal{J}_1$. \\
        \STATE Repeat Step 3-5 until $\mathcal{I}_2$ is empty. Then, recursively order the remaining $\nu$ users such that $\det{\vec{G}_{n+1}^{\nu-1}}$ is maximized at each stage. \\
         \STATE Output the user-ordering $\mathcal{U}$.\\
\end{algorithmic}
\end{algorithm}
\vspace{-4mm}

Then, we update $\mathcal{J}_2\!=\!\mathcal{J}_1\!\setminus\!\mathcal{U}(2)$, and continue to order the remaining $N\!-\!2$ users in the same way until we finish the ordering of all users. Notice that, for the last $\nu$ users, we only need to recursively select the best user that maximizes $\det{\vec{G}_{n+1}^{\nu-1}}$. Such an algorithm is summarized in Algorithm 1.

Next, we analyze the user-ordering for the minimum user-rate maximization, which renders a simpler user-ordering algorithm. From (\ref{trace4}), it holds that
\bea \sum_{n=1}^{N}\hat{g}_{n}^{\nu}\left(2^{R^{\mathrm{user}}}-1\right)\leq \frac{P_{\mathrm{T}}}{N_0}.  \eea
Therefore, the optimal user-ordering that maximizes $R^{\mathrm{user}}$ shall minimize the sum of $\hat{g}_{n}^{\nu}$, 
 \bea \mathcal{U}^{\mathrm{opt}}=\argmin_{\mathcal{U}\in\mathcal{P}}\sum_{n=1}^{N}\hat{g}_{n}^{\nu}.   \eea
As for the last user, $\hat{g}_{N}^{\nu}\!=\!g_{N, N}$ holds, we can select the user that has the smallest diagonal element $g_{n, n}$ to be the last user $\mathcal{U}(N)$. Then, for the second last user, as
  \bea \hat{g}_{N-1}^{\nu}=g_{N-1, N-1}-\frac{|g_{N-1, N}|^2}{g_{N, N}},   \eea
we can choose the user that has the second smallest diagonal element $g_{n, n}$ to be the second last user $\mathcal{U}(N\!-\!1)$. Recursively, based on (\ref{hatg}), the users can be ordered in a descending order of $g_{n, n}$, which is summarized in Algorithm 2.
\vspace{-2mm}
\begin{algorithm}
	\caption{User-ordering for minimum user-rate maximization with the GZF-DP precoder.}
	\label{alg:3}
      \begin{algorithmic}[1]
            \vspace{1mm}
      \STATE Order the user according to the descending order of the diagonal element $g_{n, n}$.
\end{algorithmic}
\end{algorithm}
\vspace{-4mm}

\section{Empirical Results}
In this section, simulation results are presented to show the promising performance of the proposed GZF-DP precoder for both the sum-rate and minimum user-rate maximizations. The sum-rate of the optimal DPC \cite{JV04} serve as the upper-bound, while the sum-rate and minimum user-rate of the linear ZF precoder serve as lower-bounds. For comparisons, we also present the rates of the UG-DP precoder in \cite{ML16} for sum-rate maximizations, which are inferior to the GZF-DP precoder with $\nu\!=\!N_g\!-\!1$ and similar DPC complexity. In all simulations, we set the noise power $N_0\!=\!1$ and test under Rayleigh fading channels that are based on the Kronecker correlation model
\bea  \vec{H}\!=\!\vec{R}_{\mathrm{R}}^{1/2}\vec{H}_{\mathrm{IID}}\vec{R}_{\mathrm{T}}^{1/2},\eea
where $N\!\times\!M$ matrices $\vec{H}_{\mathrm{IID}}$ denote IID complex Gaussian channels with zero mean and a covariance matrix being an identity matrix. The $M\!\times\!M$ matrix $\vec{R}_{\mathrm{T}}$ and $N\!\times\!N$ matrix $\vec{R}_{\mathrm{R}}$ denote the correlations at the transmit and receive sides, respectively. We use an exponential correlation model \cite{S01} for both $\vec{R}_{\mathrm{T}}$ and $\vec{R}_{\mathrm{R}}$, which is defined as
{\setlength\arraycolsep{2pt} \bea\vec{R}\!=\left[\begin{array}{ccccc}\!1&\beta&\cdots&\cdots&\beta^{K-1}\\ \beta&1&\beta&\cdots&\beta^{K-2}\\ \vdots&\ddots&\ddots&\ddots&\vdots\\ 
\vdots&\ddots&\ddots&\ddots&\beta\\
 \beta^{K-1}&\beta^{K-2}&\cdots&\beta&1\end{array}\right]\!,\eea}
\hspace*{-1.4mm}where $K\!=\!M$, $\beta\!=\!\beta_{\mathrm{T}}$ and $K\!=\!N$, $\beta\!=\!\beta_{\mathrm{R}}$ for transmit and receive correlation, respectively.

\subsection{Optimal Orderings}

In Fig. \ref{fig1}, we evaluate the sum-rate with the channel given in Example 1 for all possible $4!\!=\!24$ user-ordering schemes in $\mathcal{P}$. As it can be seen that, different user-orderings provide different sum-rate for $1\!\leq\!\nu\!\leq\!3$.

\begin{figure}[b]
\vspace*{-10mm}
\begin{center}
\hspace*{-0mm}
\scalebox{0.42}{\includegraphics{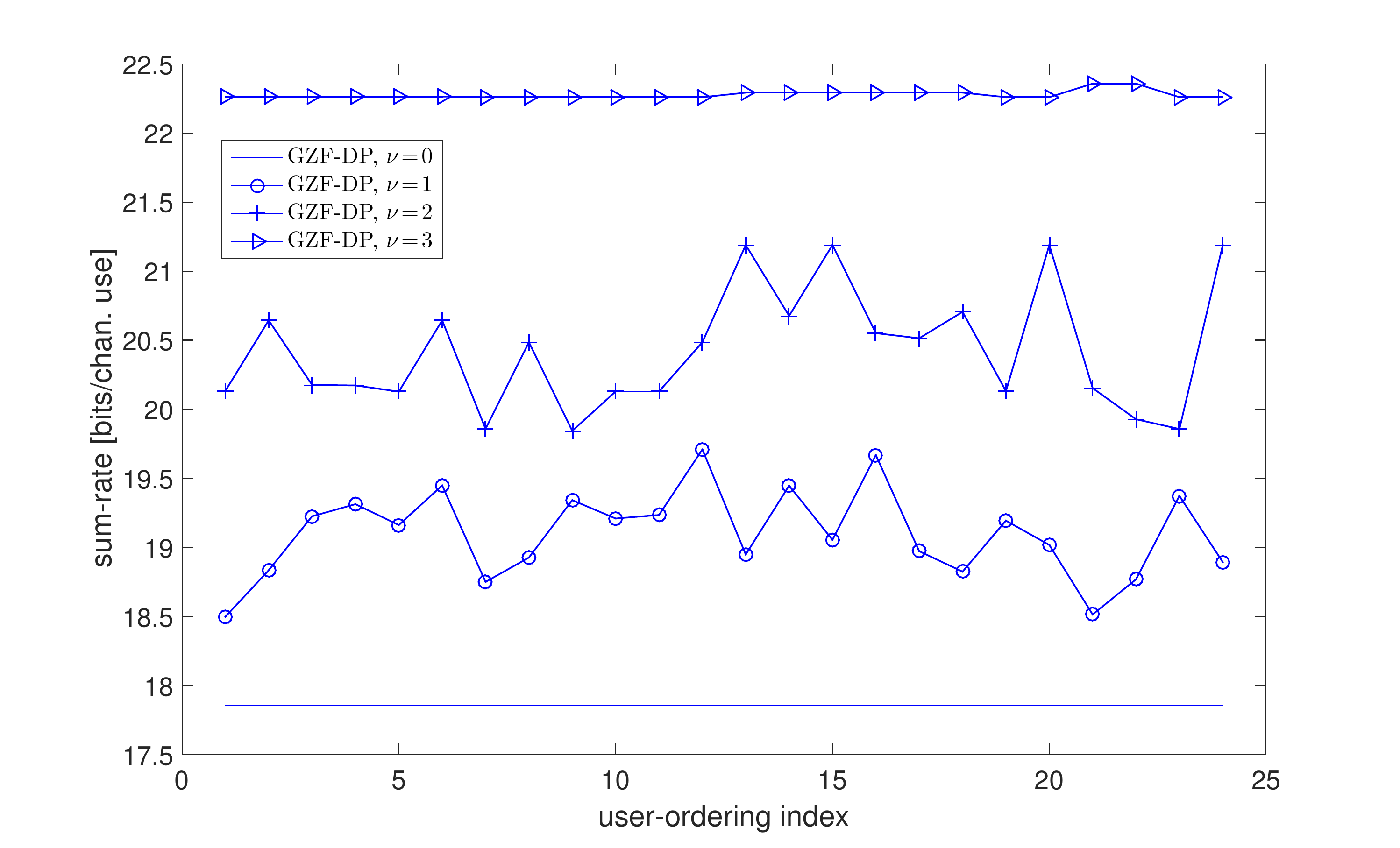}}
\vspace*{-8mm}
\caption{\label{fig1} The sum-rate of the GZF-DP precoder with different $\nu$ evaluated with $N_0\!=\!1$ and $P_{\mathrm{T}}\!=\!10$ dB. }
\vspace*{-6mm}
\end{center}
\end{figure}

In Fig. \ref{fig2}, we evaluate the performance of Algorithm 1 for user-ordering for the sum-rate maximization with $M\!=\!N\!=\!5$ and under IID complex Gaussian channels, that is, $\beta_{\mathrm{T}}\!=\!\beta_{\mathrm{R}}\!=\!0$. The optimal ordering utilizes the brute-force method to select one best user-ordering over all $5!\!=\!120$ possible combinations under each channel realization. The average sum-rate averages the sum-rate over all 120 user-orderings in $\mathcal{P}$. As can be seen, the proposed user-ordering performs 0.5 to 1 dB better than the averaged sum-rate in terms of transmit power $P_{\mathrm{T}}$.

In Fig. \ref{fig3}, we evaluate the performance of Algorithm 2 for user-ordering for the minimum user-rate maximization with $M\!=\!N\!=\!6$ and under IID complex Gaussian channels. As can be seen, the proposed Algorithm 2 performs around 1 dB better than the averaged sum-rate in terms of transmit power $P_{\mathrm{T}}$, and quite close to the optimal user-ordering that is selected over $6!\!=\!720$ possible schemes in $\mathcal{P}$ with brute-force method for each channel realization.

\subsection{Sum-rate Maximization}
Next we evaluate the sum-rate maximizations with $M\!=\!N\!=\!8$. In Fig. \ref{fig4} we simulate under IID complex Gaussian channels. As can be seen, the GZF-DP precoder with $\nu\!=\!1$ renders around 1.5 dB and 4 dB gains compared to the UG-DP precoder and the linear ZF precoder in terms of transmit power $P_{\mathrm{T}}$, respectively. With $\nu\!=\!3$, which means that in the effective channel we preserve at most interference from 3 other users for each of the users, the GZF-DP precoder is only less than 1.5 dB away from the optimal DPC, and performs quite close to the ZF-DP precoder \cite{CS03}, i.e., the GZF-DP precoder with $\nu\!=\!7$. 

In Fig. \ref{fig5}, we repeat the tests in Fig. \ref{fig4} under Rayleigh fading channels with correlation factors $\beta_{\mathrm{T}}\!=\!0.2$ and $\beta_{\mathrm{R}}\!=\!0.8$. As can be seen, the GZF-DP precoder with $\nu\!=\!1$ renders around 2 dB and 5 dB gains compared to the UG-DP and ZF precoders in this case, respectively. The $P_{\mathrm{T}}$ gains of the GZF-DP precoder are larger than those gains as in Fig. \ref{fig4}, due to the fact that the MISO broadcast channels are correlated in this case. Moreover, we also evaluate the GZF-DP precoder with user-ordering based on Algorithm 1. For the UG-DP precoder, we use the brute-force method to select the optimal user-ordering under each channel realization. As it can be seen, with user-orderings both the GZF-DP and UG-DP precoders renders higher sum-rates. But still, even with the optimal user-ordering, the UG-DP precoder is 1.5 dB away from the proposed GZF-DP precoder without user-ordering.

\begin{figure}[t]
\vspace*{-4mm}
\begin{center}
\hspace*{-0mm}
\scalebox{0.42}{\includegraphics{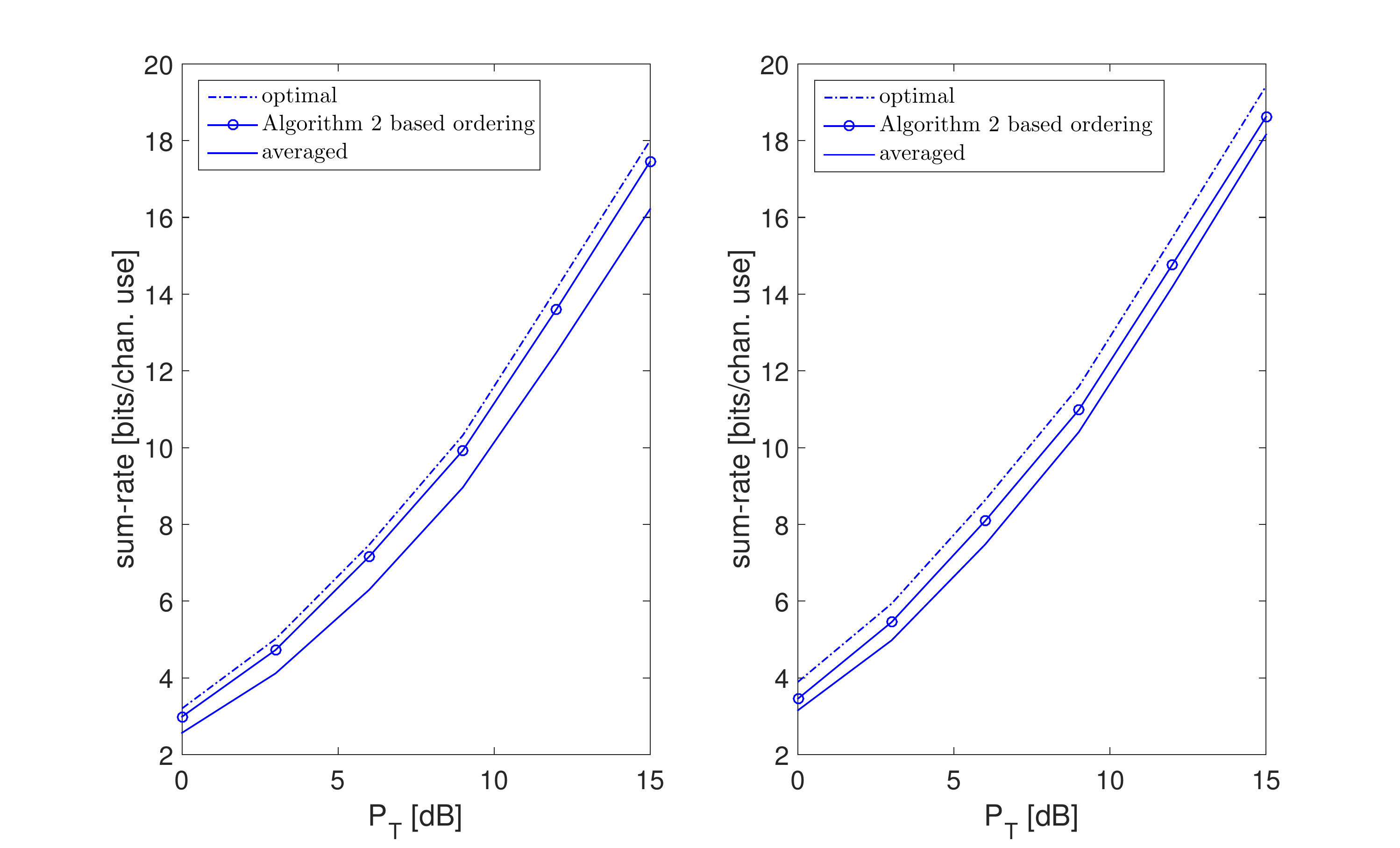}}
\vspace*{-8mm}
\caption{\label{fig2} The sum-rate of the proposed Algorithm 1 for user-ordering with the GZF-DP precoder with $\nu\!=\!1$ (the left figure) and $\nu\!=\!2$ (the right figure) for $M\!=\!N=5$.  }
\vspace*{-6mm}
\end{center}
\end{figure}

\begin{figure}
\vspace*{-6mm}
\begin{center}
\hspace*{-0mm}
\scalebox{0.42}{\includegraphics{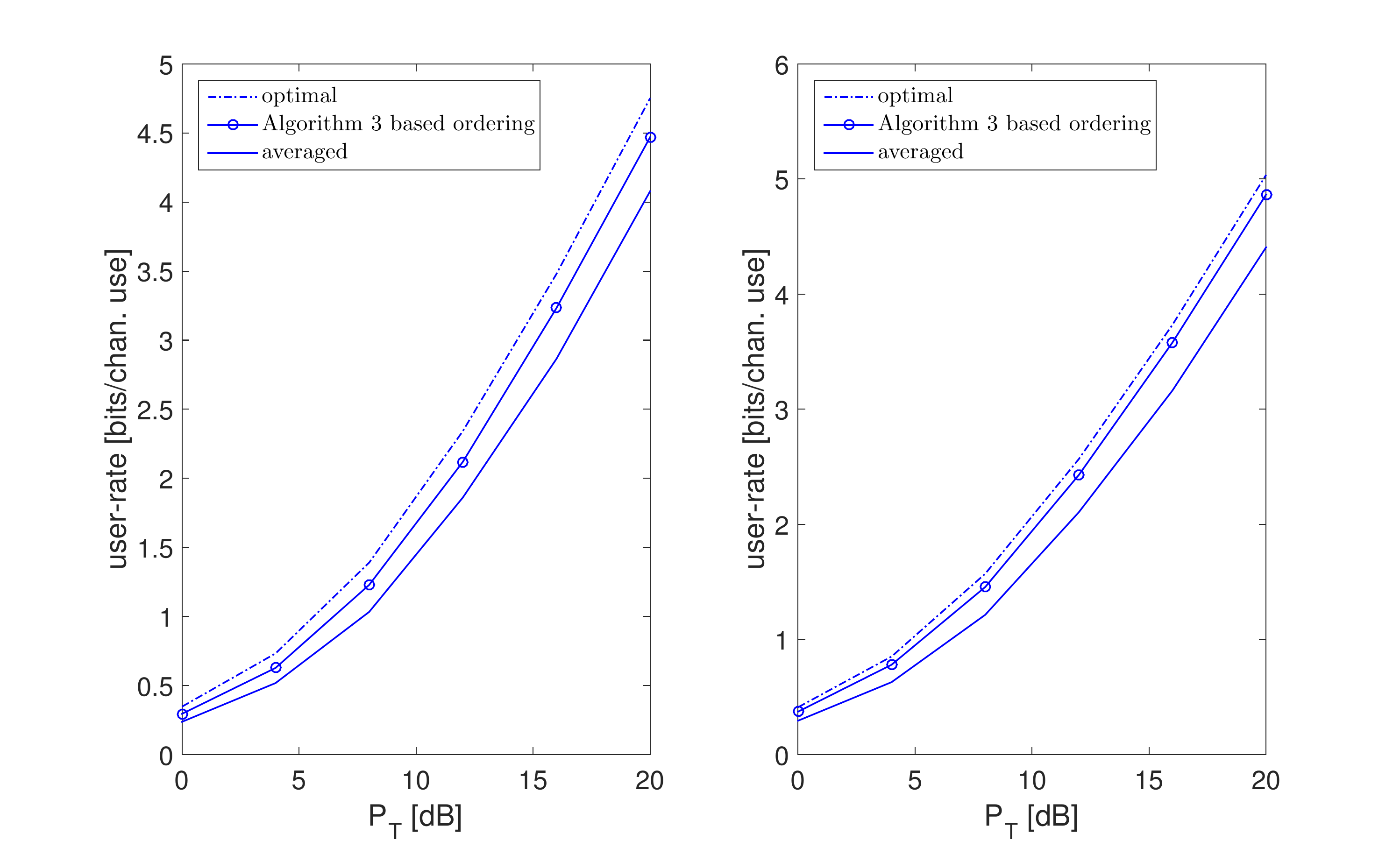}}
\vspace*{-8mm}
\caption{\label{fig3} The user-rates of the proposed Algorithm 2 for user-ordering with the GZF-DP precoder with $\nu\!=\!1$ (the left figure) and $\nu\!=\!2$ (the right figure) for $M\!=\!N\!=6$.  }
\vspace*{-10mm}
\end{center}
\end{figure}

\begin{figure}[t]
\vspace*{-4mm}
\begin{center}
\hspace*{16mm}
\scalebox{0.42}{\includegraphics{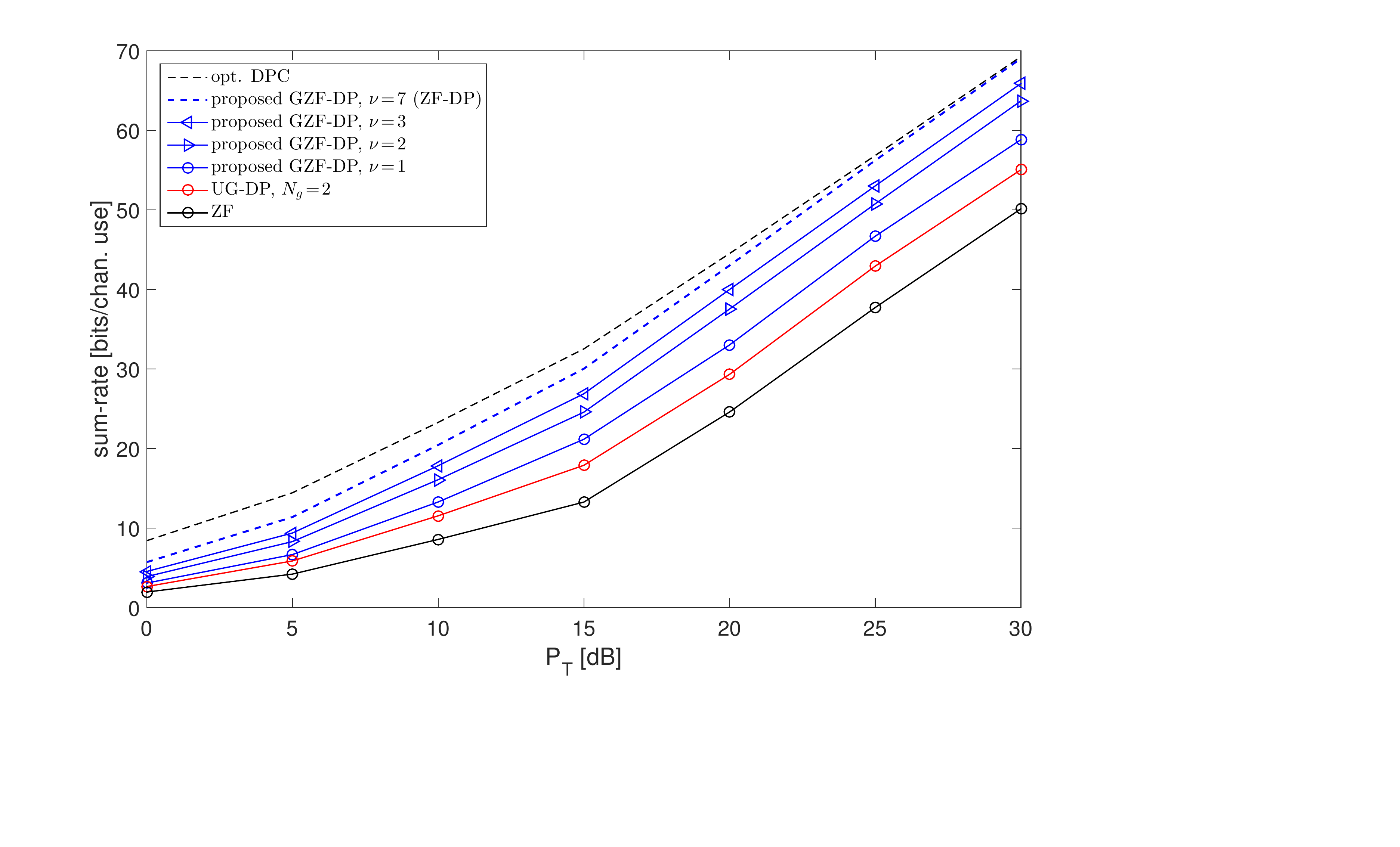}}
\vspace*{-34mm}
\caption{\label{fig4}The sum-rate maximization with $M\!=\!N\!=\!8$ under IID complex Gaussian channels.}
\vspace*{-6mm}
\end{center}
\end{figure}

\begin{figure}
\begin{center}
\vspace*{-6mm}
\hspace*{16mm}
\scalebox{0.42}{\includegraphics{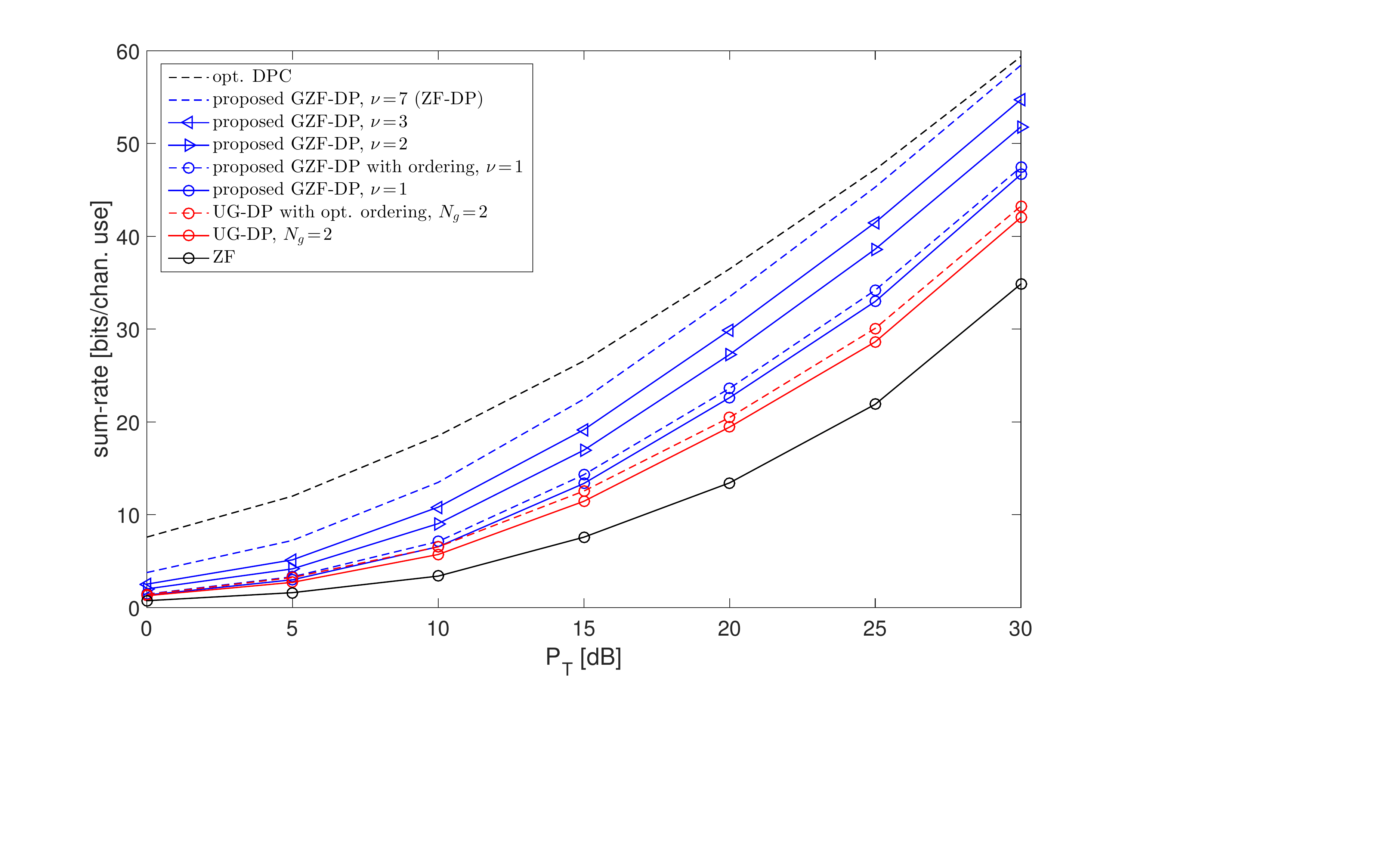}}
\vspace*{-34mm}
\caption{\label{fig5}The sum-rate maximization with $M\!=\!N\!=\!8$ under Rayleigh-fading channels and $\beta_{\mathrm{T}}\!=\!0.2$ and $\beta_{\mathrm{R}}\!=\!0.8$}
\vspace*{-10mm}
\end{center}
\end{figure}

\begin{figure}[t]
\vspace*{-4mm}
\begin{center}
\hspace*{16mm}
\scalebox{0.42}{\includegraphics{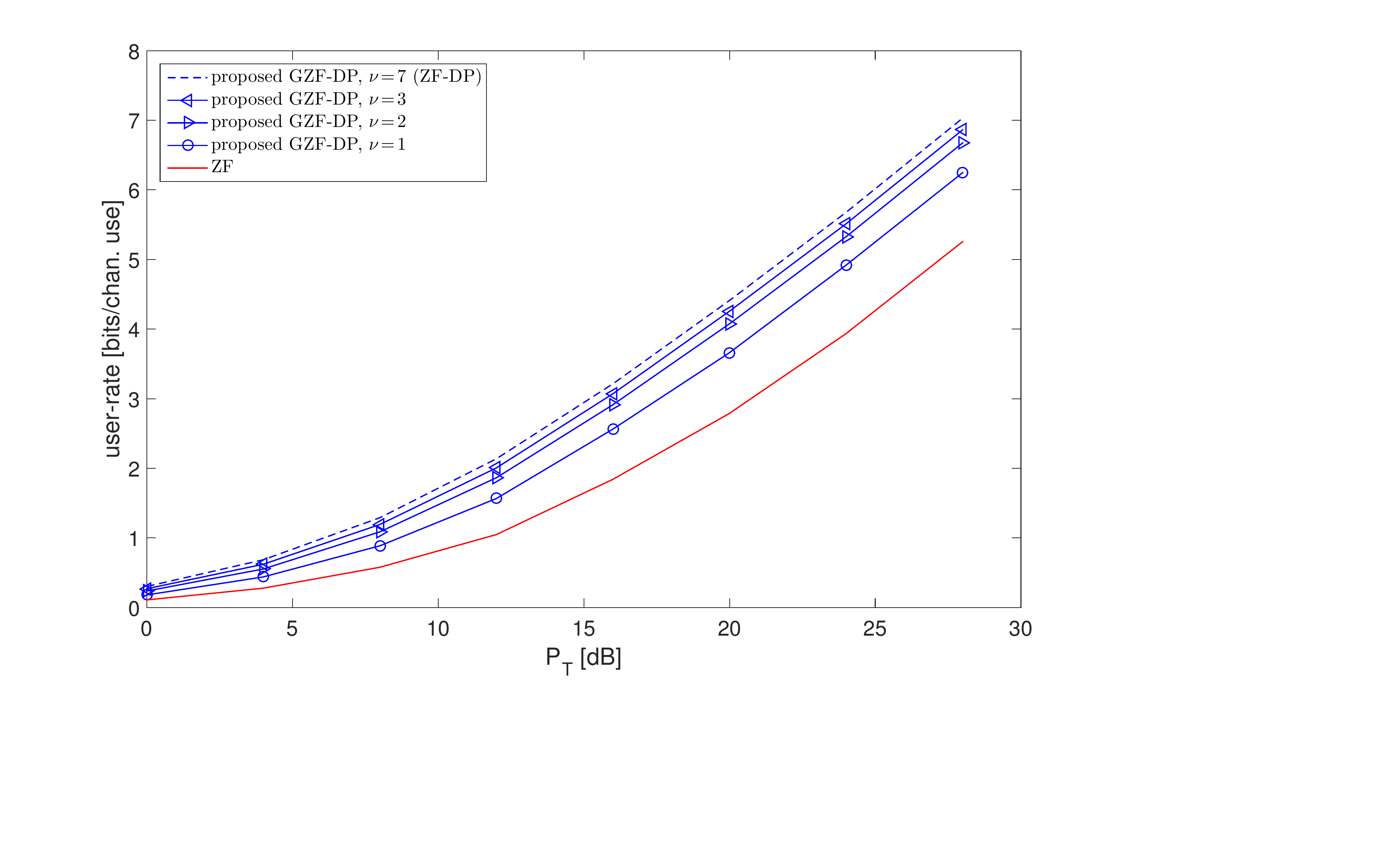}}
\vspace*{-34mm}
\caption{\label{fig6}Repeat the test in Fig. \ref{fig4} for minimum user-rate maximization under IID complex Gaussian channels.}
\vspace*{-6mm}
\end{center}
\end{figure}

\begin{figure}
\vspace*{-6mm}
\begin{center}
\hspace*{0mm}
\scalebox{0.42}{\includegraphics{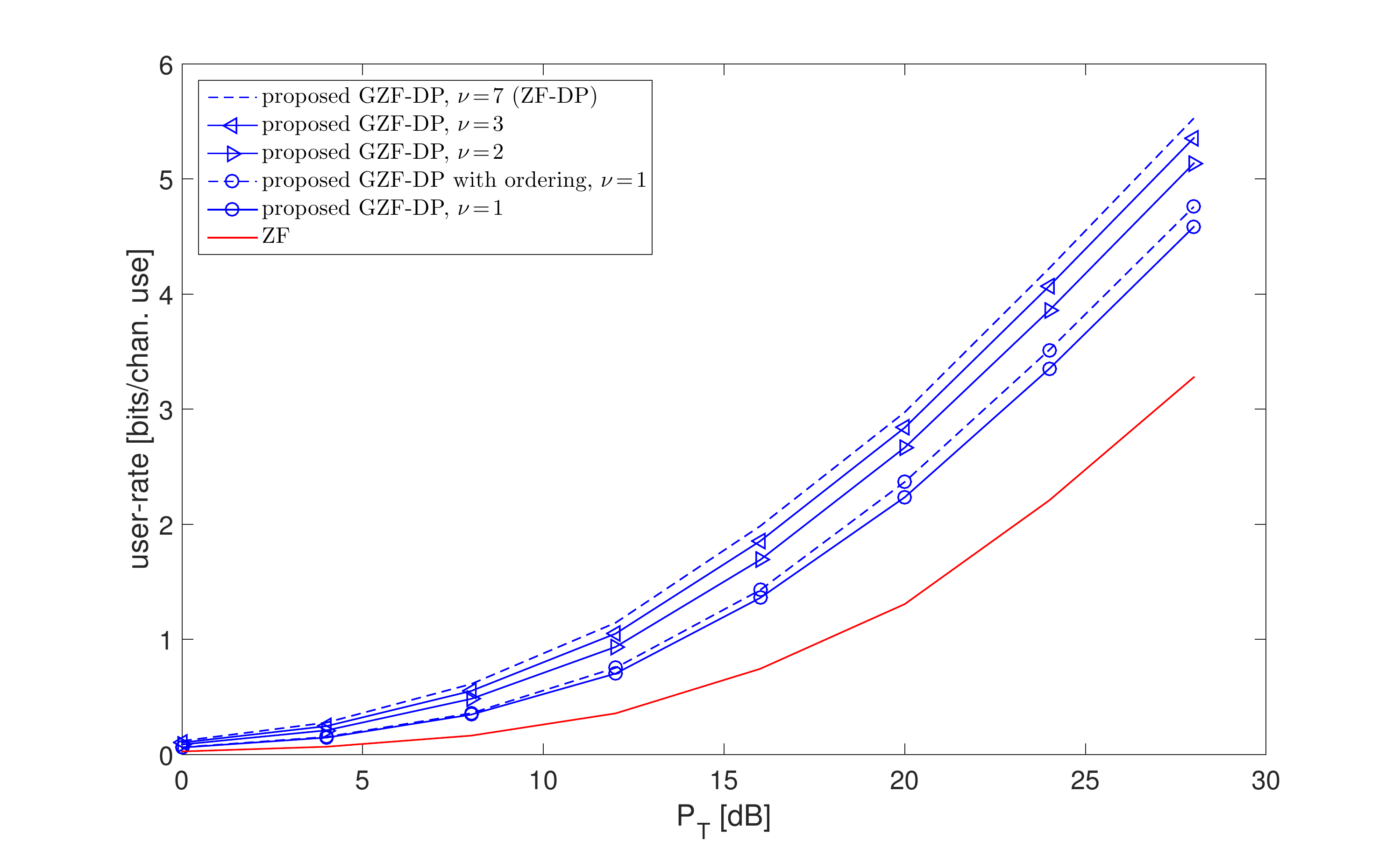}}
\vspace*{-8mm}
\caption{\label{fig7}Repeat the test in Fig. \ref{fig5} for minimum user-rate maximization under Rayleigh-fading channels.}
\vspace*{-10mm}
\end{center}
\end{figure}

\subsection{Minimum User-rate Maximization}
Next we evaluate the minimum user-rate maximizations with $M\!=\!N\!=\!8$ and repeat the tests in Fig. \ref{fig4} and Fig. \ref{fig5}, respectively. As can be seen, in Fig. \ref{fig6} the proposed GZF-DP precoder with $\nu\!=\!1$ is around 2 dB better than the linear ZF precoder, while in Fig. \ref{fig7} the gain is more than 4 dB due to spatial correlated channels. In addition, in both cases, the GZF-DP precoder with $\nu\!=\!3$ performs close to the GZF-DP precoder with $\nu\!=\!7$, i.e., the ZF-DP precoder.

\begin{figure}[t]
\vspace*{-4mm}
\begin{center}
\hspace*{0mm}
\scalebox{0.42}{\includegraphics{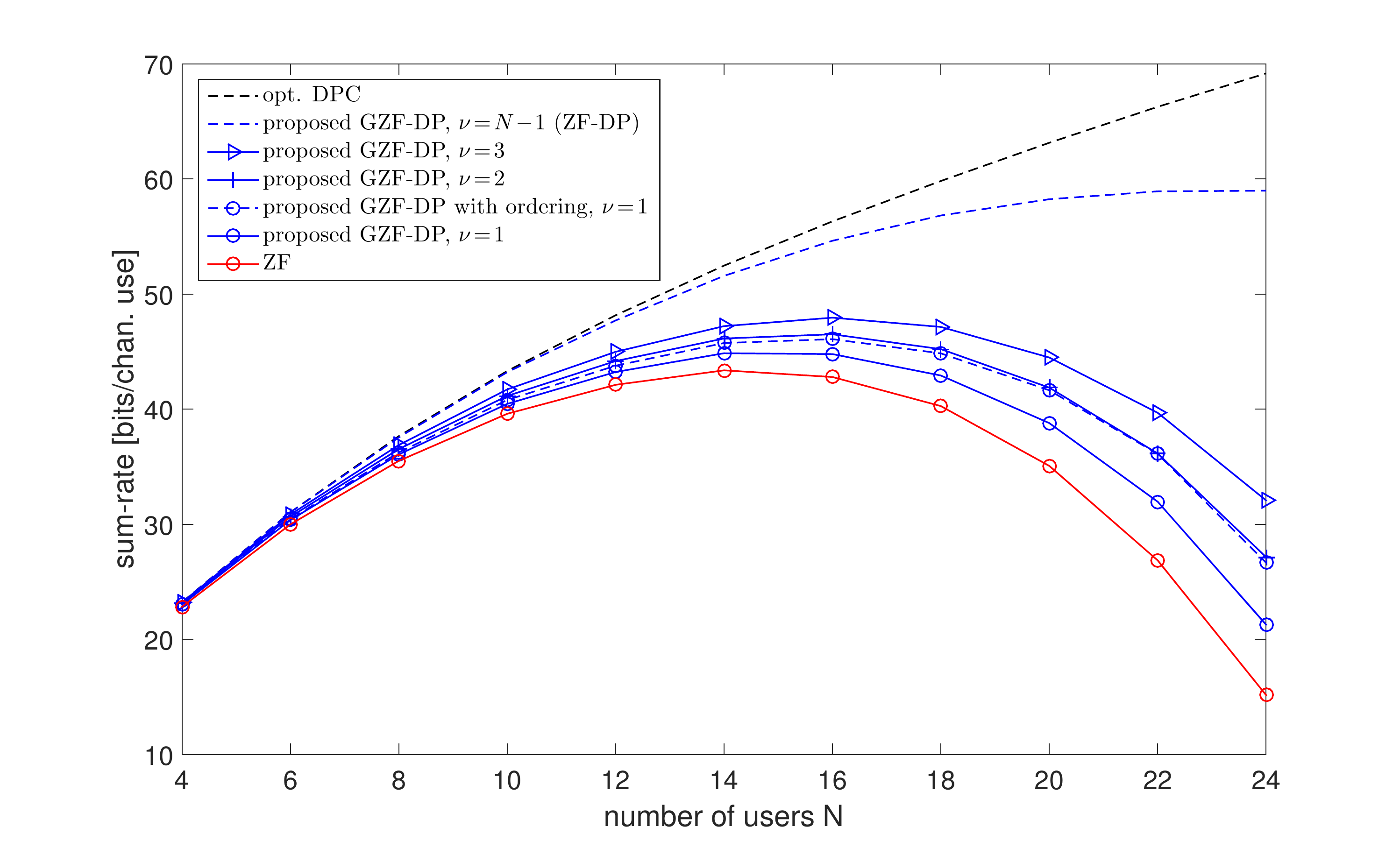}}
\vspace*{-8mm}
\caption{\label{fig8}The sum-rate maximization with $M\!=\!24$ and different number of users $N$ under IID complex Gaussian channels. Note that the transmit power is constant no matter the number of users $N$.}
\vspace*{-6mm}
\end{center}
\end{figure}

\begin{figure}
\vspace*{-6mm}
\begin{center}
\hspace*{16mm}
\scalebox{0.42}{\includegraphics{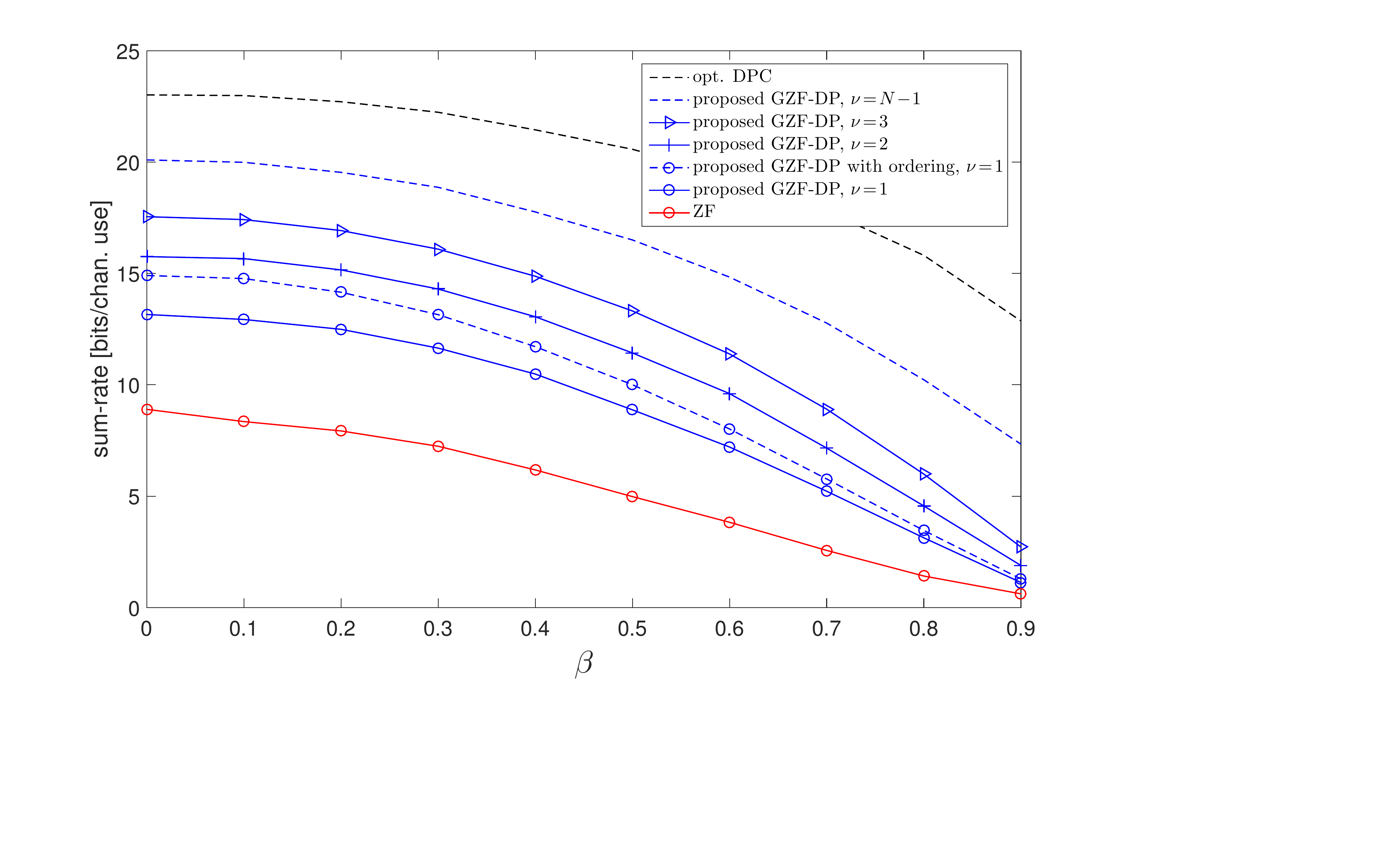}}
\vspace*{-34mm}
\caption{\label{fig9}The sum-rate maximization with $M\!=\!N\!=\!8$ under Rayleigh-fading channels. The correlation factors $\beta_{\mathrm{T}}\!=\!\beta_{\mathrm{R}}$, and change from 0.1 to 0.9.}
\vspace*{-10mm}
\end{center}
\end{figure}

\begin{figure}[t]
\vspace*{-4mm}
\begin{center}
\hspace*{0mm}
\scalebox{0.42}{\includegraphics{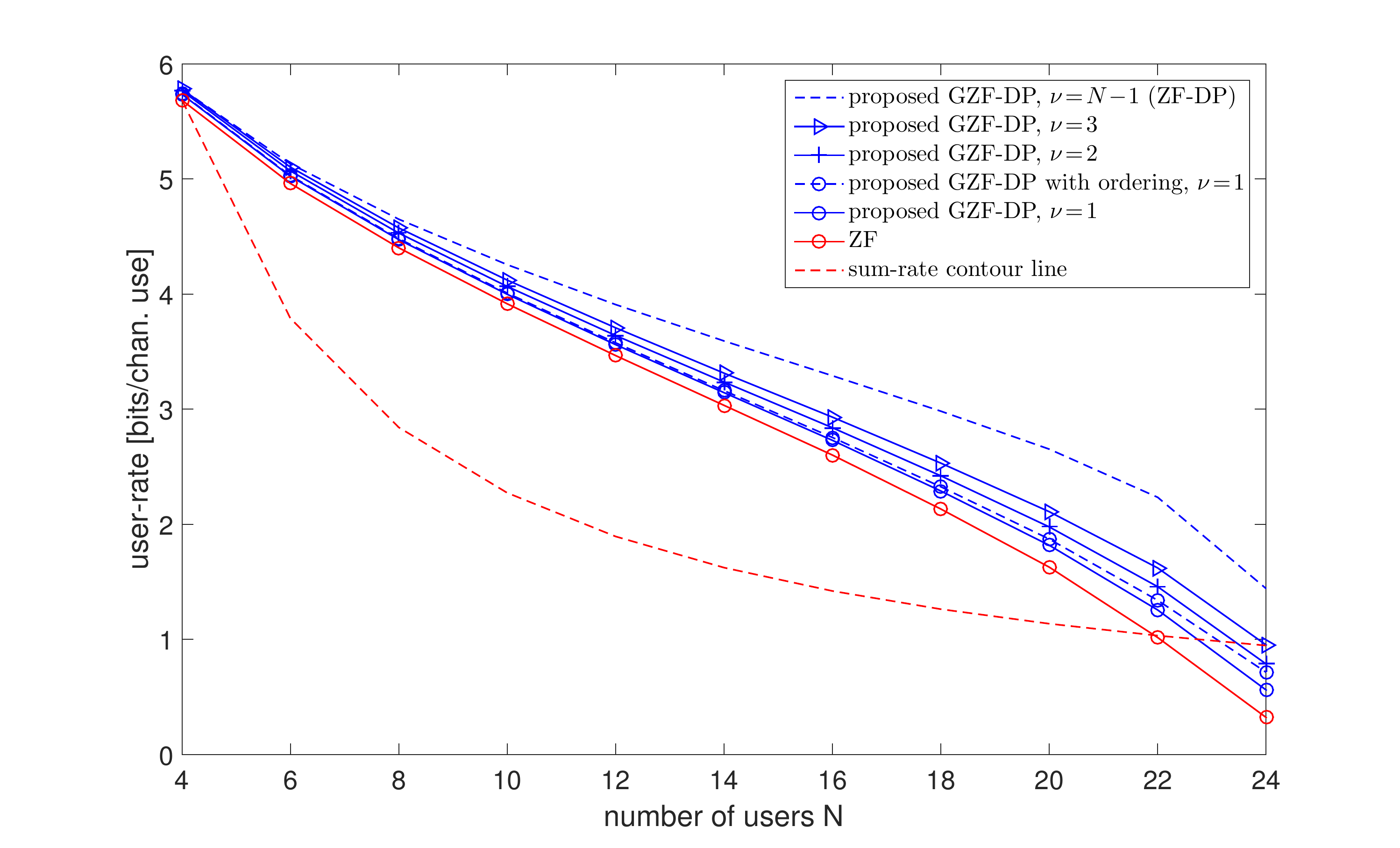}}
\vspace*{-8mm}
\caption{\label{fig10}Repeat the tests in Fig. \ref{fig8} for user-rate maximization.}
\vspace*{-6mm}
\end{center}
\end{figure}

\begin{figure}
\vspace*{-6mm}
\begin{center}
\hspace*{0mm}
\scalebox{0.42}{\includegraphics{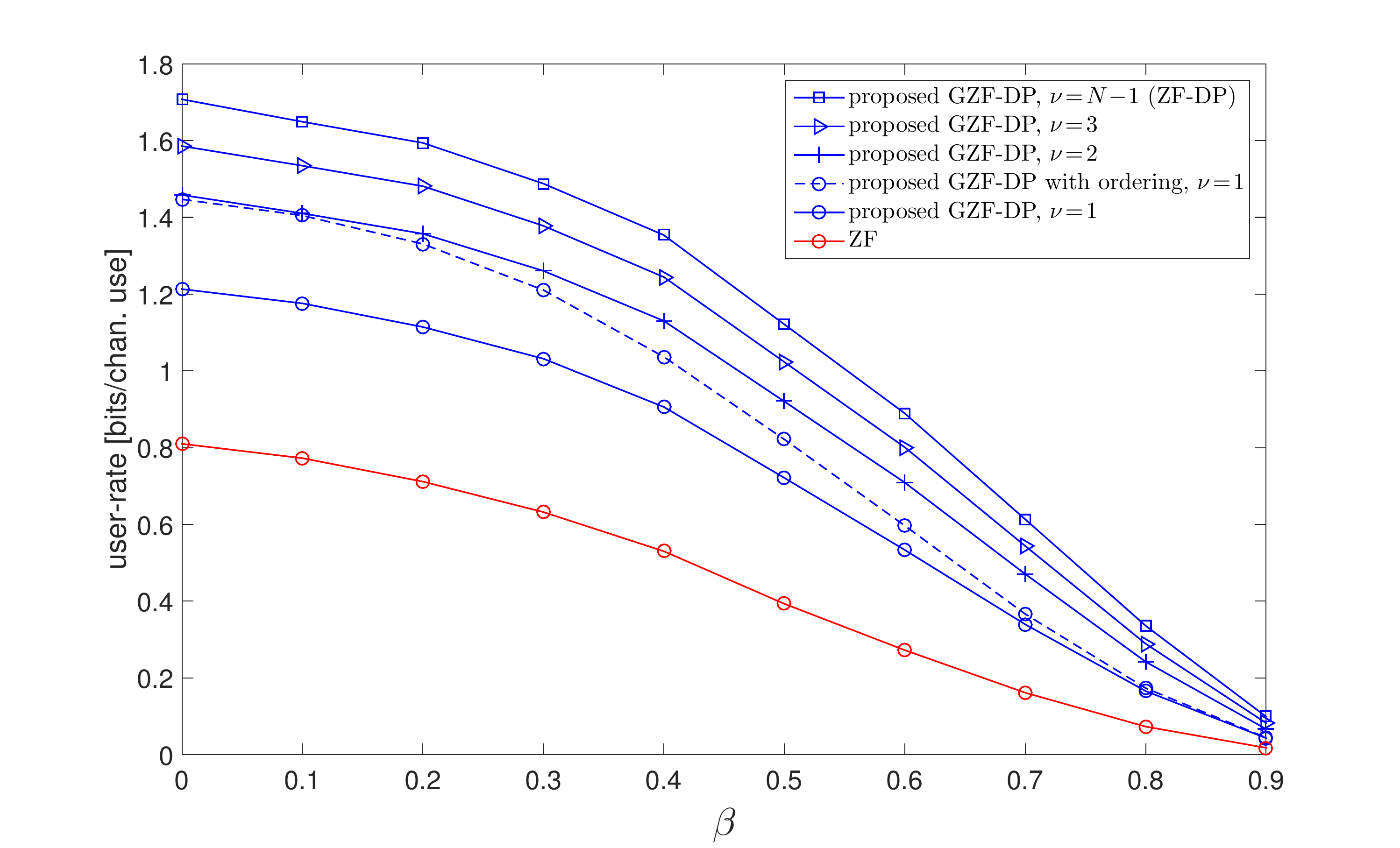}}
\vspace*{-8mm}
\caption{\label{fig11}Repeat the tests in Fig. \ref{fig9} for user-rate maximization.}
\vspace*{-10mm}
\end{center}
\end{figure}

\subsection{Impact of the Number of Users and Correlation Factors}
Next we evaluate the impacts of increasing the number of users and the spatial correlation factors. In all simulations, we set the total transmit power $P_{\mathrm{T}}\!=\!10$ dB. In Fig. \ref{fig8}, we set the number of transmit antennas $M\!=\!24$ and increase the user number $N$ from 4 to 24. As can be seen, as the number of users increases, the sum-rate first increases and then decreases both for the linear ZF precoder and the GZF-DP precoder with $\nu\!<\!N\!-\!1$. This is so, since as $N$ increases the degrees of freedom (DoF) for the precoder designs also increase and consequently the sum-rate is getting higher. However, the inter-user interference increased with a larger $N$ causes sum-rate degradation for small values of $\nu$. Nevertheless, it can be seen that the GZF-DP precoder with $\nu\!=\!1$ renders the same sum-rate as the ZF precoder with one user less.

In Fig. \ref{fig9}, we set $M\!=\!N\!=\!8$ and $\beta_{\mathrm{T}}\!=\!\beta_{\mathrm{R}}\!=\!\beta$. We increase $\beta$ from 0.1 to 0.9. As can be seen, as $\beta$ gets higher, the sum-rate decreases for all precoders. At low and medium correlations, the GZF-DP precoder shows significant gains over the linear ZF precoder. For instance, the GZF-DP precoder with $\beta\!=\!0.5$ renders the same sum-rate as the linear ZF precoder with $\beta\!=\!0$. Therefore, the GZF-DP precoder is more robust against the transmit and receive correlations compared to the linear ZF precoder. In addition, with the user-ordering proposed in Algorithm 1 the correlation gain is even larger.

In Fig. \ref{fig10}, we repeat the tests in Fig. \ref{fig8} for minimum user-rate maximizations. As can be seen, unlike the cases of the sum-rate maximizations, as the number of users increases, the user-rates of all precoder designs decrease. We also present a contour line of the sum-rate, which shows that the sum-rate also decreases when $N$ is close to $M$. For large $N$, we can see that the GZF-DP precoder with $\nu\!=\!1$ renders the same user-rate as the linear ZF precoder with one user less.

In Fig. \ref{fig11}, we repeat the tests in Fig. \ref{fig9} for minimum user-rate maximizations. As can be seen, as the correlation factor $\beta$ gets higher, the user-rates also decrease for all precoders. The GZF-DP precoder again shows superior performance compared to the linear ZF precoder, and is more robust against transmit and receive correlations.

\subsection{Practical FD-MIMO Scenario}
At last, we evaluate the proposed GZF-DP precoder in an FD-MIMO downlink scenario considering a 3D channel model \cite{KJ03}. The test scenario is depicted in Fig. \ref{fig12}, where we have an $8\!\times\!8$ 2D antenna-array deployed at an e-NodeB that is 20 meters above the ground. The spacing between to adjacent antenna elements (both in horizontal and vertical dimensions) is 1/2 wave-length. The e-NodeB broadcasts at 2.4 GHz to 8 single-antenna users that are placed along a line which is perpendicular to the 2D antenna-plane. The distance between two adjacent users is 10 meters and the first user is 20 meters away from the e-NodeB. For simplicity, we consider an ideal line-of-sight (LOS) situation with channels constructed from the free-space path loss.

As shown in Fig. \ref{fig13}, the sum-rate of the proposed GZF-DP precoder with $\nu\!=\!1$ is much higher than that of the linear ZF precoder. And with $\nu\!=\!3$ the GZF-DP precoder significantly outperforms the UG-DP precoder with $N_g\!=\!4$. Moreover, the GZF-DP precoder with $\nu\!=\!3$ also performs close to the ZF-DP precoder which requires a full successive DPC scheme.

\begin{figure}[t]
\vspace*{-4mm}
\begin{center}
\hspace*{0mm}
\scalebox{0.65}{\includegraphics{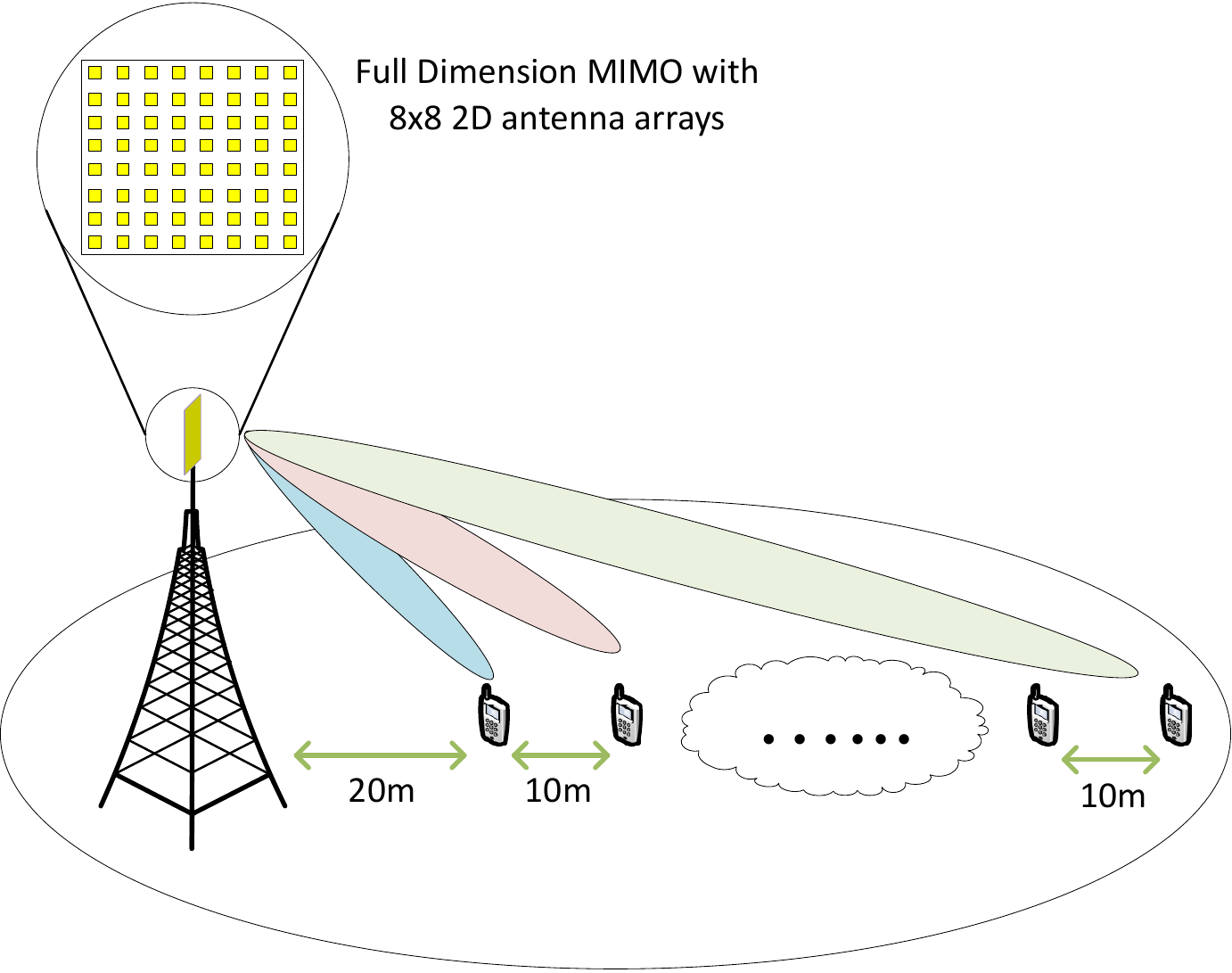}}
\vspace*{-5mm}
\caption{\label{fig12}An FD-MIMO scenario where an e-NodeB equipped with an $8\!\times\!8$ 2D antenna-array is broadcasting to 8 lined-up single-antenna users.}
\vspace*{-6mm}
\end{center}
\end{figure}

\begin{figure}
\vspace*{-6mm}
\begin{center}
\hspace*{16mm}
\scalebox{0.42}{\includegraphics{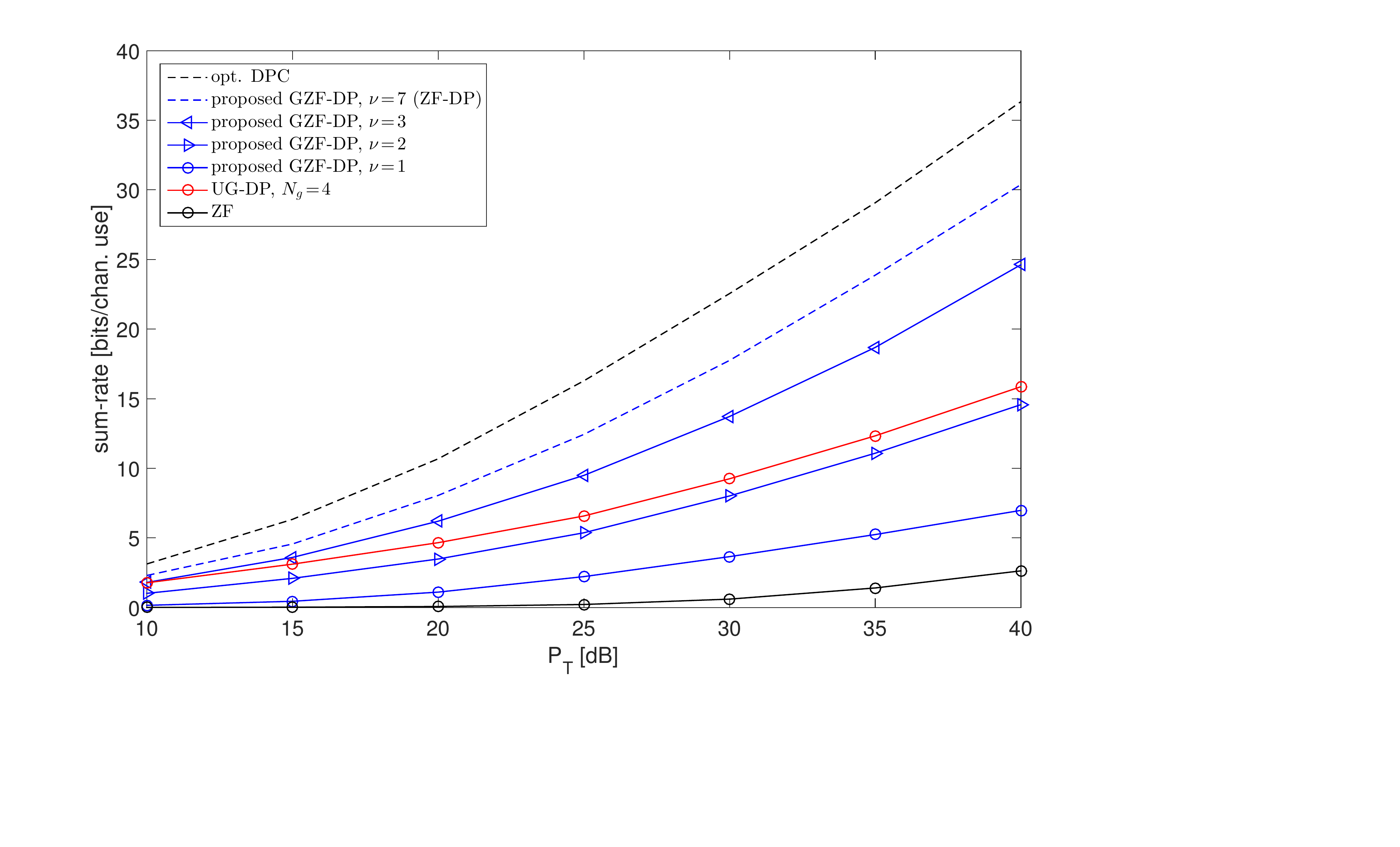}}
\vspace*{-34mm}
\caption{\label{fig13}The sum-rate maximization for the FD-MIMO scenario considered in Fig. \ref{fig12}.}
\vspace*{-10mm}
\end{center}
\end{figure}

\section{Summary}
We have proposed a generalized zero-forcing precoder (GZF) in conjunction with successive dirty-paper coding (DPC), namely, the GZF-DP precoder, for multi-input-single-output (MISO) broadcast channels. Utilizing the successive DPC encoding scheme at the transmitter to cancel the known non-causal interference, the GZF-DP precoder preserves up to $\nu$ interferers for each of the users and results in significant rate-increments. We analyze optimal designs of the proposed GZF-DP precoder both for sum-rate and minimum user-rate maximizations. The optimal GZF-DP precoder designs are solved in closed-forms in relation to optimal power allocations. For the sum-rate maximization, the optimal power allocation can be efficiently found with modified water-filling schemes introduced by inter-user interference, while for the minimum user-rate maximization, the optimal power allocation is solved in closed-from. We have also derived two efficient and low-complexity user-ordering algorithms for the GZF-DP precoder for the sum-rate and minimum user-rate maximizations, respectively. We show through numerical results that, the proposed GZF-DP precoder yields both much higher sum-rate and minimum user-rate compared to the traditional linear ZF precoder and the previous user-grouping based DPC (UG-DP) precoder, and is close to the ZF with full complexity DPC (ZF-DP) precoder.

\end{document}